\documentclass{article}
\usepackage{graphicx} 
\usepackage{amsmath,amsthm,soul}
\usepackage{stmaryrd,graphicx,scalerel}
\usepackage{amssymb}
\usepackage{amsbsy}
\usepackage{graphics}
\usepackage{pifont}
\usepackage{epsfig}
\usepackage{setspace}
\usepackage{graphicx}
\usepackage{enumerate}
\usepackage[english]{babel}
\usepackage[square,numbers]{natbib}
\usepackage{caption}

\usepackage{multicol}
\usepackage{epstopdf}
\usepackage{wrapfig}
\usepackage{subfigure}
\usepackage{setspace}
\usepackage{svg}
\usepackage{xcolor,calc}
\usepackage{siunitx}  
\usepackage{booktabs} 
\usepackage{graphicx} 

\newcommand{\bs}[1]{\boldsymbol{#1}}

\newtheorem{lemma}{Lemma}

\title{Structure-Informed  Neural Networks for Boundary Observation Problems}
\author{Jakub Horsky, Andrew Wynn }
\date{April 2023}

\begin{document}

\maketitle

\begin{abstract}
    We introduce \textit{Structure Informed Neural Networks} (SINNs), a novel method for solving boundary observation problems involving PDEs. The SINN methodology is a data-driven framework for creating approximate solutions to internal variables on the interior of a domain, given only boundary data. The key idea is to use neural networks to identify a co-ordinate transformation to a latent space, upon which a well-posed elliptic system of partial differential equations is constructed. The use of elliptic systems enables the low-cost transfer of information from the domain's boundary to its interior. This enables approximate solutions to PDE boundary observation problems to be constructed for generic, and even ill-posed, problems. A further advantage of the proposed method is its ability to be trained on experimental or numerical data without any knowledge of the underlying PDE. We demonstrate the ability of SINNs to accurately solve boundary observation problems by considering two challenging examples of a non-linear heat equation and boundary observation for the  Navier-Stokes equations. 

\textit{Key words:} Data driven scientific computing, Reduced order modeling, Machine learning, Partial differential equations, Operator learning
\end{abstract}

\section{Introduction}

Boundary observation problems aim to discover the value of a physical quantity inside a domain by using only observations from its boundary. If possible, this means that potentially complex physical information can be obtained without the need for invasive internal sensors. Many fundamental problems in engineering and physics have this form with applications, for example, in  fluid mechanics \cite{illingworth_2018}, medical imaging \cite{song_2022}, geophysics \cite{snieder_1988}, and thermal sensing \cite{bryan_2023}.  

Typically, the internal physical quantity of interest is linked to the boundary observations by a  partial differential equation (PDE). In many applications, this can make the problem highly challenging to analyse analytically and computationally impractical to solve numerically. In fluid mechanics, for example, the nonlinear Navier-Stokes equations govern the relation between the internal fluid properties, such as its velocity or temperature, and boundary data which are convenient to observe experimentally, such as pressure or shear stress. The well-known complexity of solutions to such nonlinear PDEs implies that solving boundary observation problems of practical importance is a significant challenge. 

In this paper, we propose a new data-driven methodology, called {\em Structure Informed Neural Networks} (SINNs), for solving boundary observation problems involving nonlinear PDEs. The idea is to embed an inherently well-posed structure for boundary observation problems into a data-driven framework with the aim to enable efficient, low-order, approximate solutions. This is achieved in a three-stage process, indicated schematically in Figure \ref{fig:mapping_schematic}. First, a neural network encodes both the boundary data and the structure of the boundary geometry into a simpler {\em latent space} of boundary variables. Information is then passed from the boundary to the interior of the latent space using an {\em elliptic system} \cite{giaquinta_2012}. This embeds a general class of well-posed PDEs into the SINN. Finally, a second neural network is used to decode the interior latent to physical variables. 

The idea of using elliptic systems in a data-driven approach is the main novelty of this paper. Boundary value problems for elliptic systems were widely studied in the ``golden age'' of PDE analysis in the 1950s \cite{morrey_1957}. Our motivation for using them now is that they can describe a significant range of boundary value problems, are numerically tractable to solve, and can be defined with only small number of parameters.  The second major contribution of this paper is to  develop an operator-theoretic framework for embedding elliptic systems within the classical encoder-decoder structure of neural network-based reduced order modelling. This underpins the efficient numerical identification of SINNs, enables a powerful coupling of elliptic systems with deep neural networks, and opens the door to the data-driven solution of a wide range of challenging nonlinear boundary observation problems. 

The structure-informed neural networks (SINNs) developed here have some similarities, and take inspiration from, a number of existing data-driven methods for PDE analysis. For example, Koopman-based modal decomposition methods \cite{bevenda_2021,schmid_2022,Lusch_2018} possess the same three-stage mapping structure as in Figure \ref{fig:mapping_schematic};  Physics-Inspired Neural Networks (PINNs) \cite{raissi_2019} use neural networks to efficiently solve PDEs, including boundary value problems \cite{li_2023}; and Neural Operators \cite{Kovachki_2023} use a kernel-based Neural Networks to construct solution operators for PDE parameter identification.  To enable a full discussion of the relation and distinction between SINNs and existing methods in \S \ref{sec:SINNS_vs_others} we must first define the mathematical structure of the boundary observations problems we aim to solve and give an overview of the SINN methodology.

\begin{figure}
\centering
\includegraphics[width=0.95\textwidth,clip=true,trim=1cm 1cm 1cm 1cm]{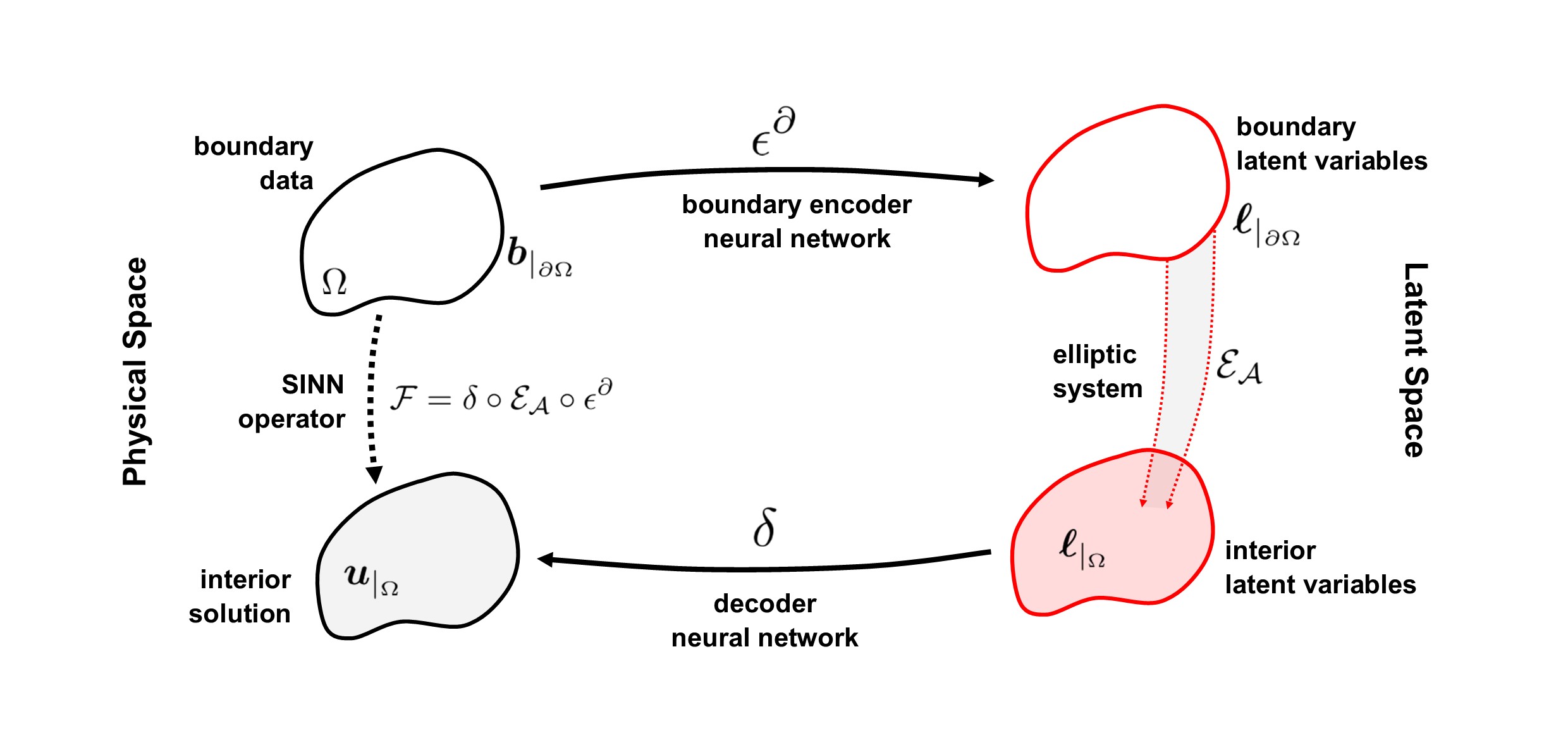}    
\caption{Schematic overview of a Structure-Informed Neural Network (SINN) operator $\mathcal{F} = \delta \circ \mathcal{L} \circ \epsilon^{\partial}$ for solving boundary observation problems.}
\label{fig:mapping_schematic}
\end{figure}

\subsection{Boundary observation problems}
Consider a physical domain $\Omega \subset \mathbb{R}^d$, with $d=2$ or $3$, and let $\partial \Omega$ denote its boundary. At each point $\bs{x} \in \Omega$, we want to recover the value of an $n$-dimensional physical variable $\bs{u}(\bs{x}) \in \mathbb{R}^n$. To do this, we can only make use of boundary data $\bs{b}(\bs{z}) \in \mathbb{R}^{n_\partial}$ which can be measured at each point $\bs{z} \in \partial \Omega$. It will be assumed that both interior and boundary data are square-integrable functions in the sense that $\bs{u} \in X$, where 
\[
X = L^2(\Omega,\mathbb{R}^n) = \left\{ f: \Omega \rightarrow \mathbb{R}^n : \int_\Omega \|f(\bs{x})\|^2 d\bs{x} < \infty \right\},
\]
and that $\bs{b} \in Y$, where 
\[
Y = L^2(\partial \Omega, \mathbb{R}^{n_\partial}) = \left\{ f: \partial\Omega \rightarrow \mathbb{R}^{n_\partial} : \int_{\partial \Omega
} \|f(\bs{z})\|^2 d\bs{z} < \infty \right\}.
\]

A typical situation in which such data arises is if the internal and boundary data satisfy a PDE of the form 
\begin{equation} \label{eq:PDE_abstract}
\begin{split}
\mathcal{L}(\bs{u},\bs{\lambda}) &= 0, \qquad \text{in} \; \Omega\\    
\mathcal{B}(\bs{u}) &= \bs{b}, \qquad  \text{on} \; \partial\Omega,  
\end{split}
\end{equation}
where $\mathcal{L}$ is a differential operator, $\bs{\lambda}$ are any parameters, and $\mathcal{B}$ is an output operator linking the interior to boundary variables. We do not assume that \eqref{eq:PDE_abstract} is a well-posed in the sense that for any boundary function $\bs{b}$, there is a unique solution $\bs{u}$ satisfying the PDE. Instead, the output operator $\mathcal{B}$ should be viewed simply as shorthand for the ``available information'' which may be observed on the boundary, given that a physical system is in state $\bs{u}$ inside the domain.

In this abstract language, the structure-informed neural networks (SINNs) that will be constructed in this paper are  operators   
\begin{equation} \label{eq:SINN_abstract}
\begin{split}
\mathcal{F} :& Y \rightarrow X \\
& \bs{b} \mapsto \bs{u}
\end{split}
\end{equation}
acting between the function space $Y$ of observable measurements and the space $X$ of all possible distributions of interior physical values. The fact we will identify {\em operators} is important since it means that a single SINN  $\mathcal{F}$ is able to  approximate the internal variables $\bs{u}$, given any possible boundary observation $\bs{b} \in X$. As will be discussed in \S \ref{sec:SINNS_vs_others}, this operator-based philosophy places SINNs within the recent class of neural-network-based operator identification methods such as Neural Operators \cite{Kovachki_2023} or DeepONets \cite{lu_2021}. 

A key objective of this paper is to identify SINN operators from data. We  assume that a data ensemble 
\[
\mathcal{U}:=\{ (\bs{u}_i , \bs{b}_i) \}_{i=1}^{N_T} \subset (X \times Y)^{N_T}
\]
is available consisting of $N_T$ pairs of internal and boundary data, each arising from a solution to \eqref{eq:PDE_abstract}. The idea will be to construct a mapping of the form \eqref{eq:SINN_abstract} which optimally fits the data $\mathcal{U}$. Given the infinite-dimensional nature of the underlying problem and the finite-dimensional nature of the data $\mathcal{U}$, however, any numerically tractable method must make an a-priori restriction on the possible forms that the $\mathcal{F}$ can take.

\subsection{SINN operators} \label{sec:intro_SINN_operators}

We assume that $\mathcal{F}$ is the composition of three operators
\begin{equation} \label{eq:F_composition}
\mathcal{F} = \delta \circ \mathcal{E}_\mathcal{A} \circ \epsilon^{\partial},
\end{equation}
the structure of which is shown schematically in Figure \ref{fig:mapping_schematic}. The first operator is called a {\em boundary encoder} $\epsilon^\partial : Y \rightarrow Y_L$. This is a nonlinear operator, defined in terms of a neural network, which maps both the boundary data and geometry into a boundary latent space, $Y_L :=L^2(\partial \Omega, \mathbb{R}^r)$, where the parameter $r$ governs the order and complexity of the latent space. 

To enable data-driven training we must further restrict the form of the operator $\epsilon^\partial$, and we assume that $\epsilon^\partial$ acts {\em semi-locally} in the following sense. Given boundary data $\bs{b} \in Y$, the value of $(\epsilon^\partial \bs{b})(\bs{z})$ at any $\bs{z} \in \partial \Omega$ can only depend on the values of $\bs{b}$ in a small neighbourhood $N_{\bs{z}} \subset \partial \Omega$ of $\bs{z}$. Practically, this will be achieved by training a neural network\footnote{A formal mathematical definition of neural networks used in this paper is given in \S\ref{sec:cost_min}.} $\mathcal{N}^\partial: \left\{ \bs{b}(\bs{\xi}): \bs{\xi} \in N_{\bs{z}} \right\} \mapsto (\epsilon^\partial \bs{b})(\bs{z})$ As will be described in detail in \S\ref{sec:bdry_encoder}, the fact that the input to $\mathcal{N}_\partial$ is defined in terms of a local neighbourhood will enable the use of a single neural network to be repeatably to build up the definition $\epsilon^\partial : Y \rightarrow Y_L$. This enables a wide class of nonlinear operators to be considered without significantly increasing the number of optimisation parameters.

The purpose of introducing latent variables is to define a common structure within which  information can be passed from the boundary latent space $Y_L$ to an interior latent space $X_L=L^2(\Omega,\mathbb{R}^r)$. A SINN implements this transfer of information by using an {\em elliptic system} of PDEs. An elliptic system is governed by a second-order differential operator
\[
D_\mathcal{A} \bs{\ell} = \sum_{i,j=1}^d A_{ij} \frac{\partial^2 \bs{\ell}}{\partial x_i \partial x_j},
\]
where $A_{ij} \in \mathbb{R}^{r \times r}$ are symmetric matrices satisfying the two conditions: i) that $A_{ij}=A_{ji}$, for any $i,j=1,\dots,r$; and ii) that the block matrix $\mathcal{A} = (A_{ij}) \in \mathbb{R}^{rd \times rd}$ is strictly positive definite.

Information is passed from the boundary latent space $Y_L$ to an interior latent space $X_L$ by solving the following boundary problem:
\begin{equation} \label{eq:elliptic_system}
\begin{split}
D_\mathcal{A} \bs{\ell} &= 0, \phantom{(b)^\partial}  \qquad \text{in}\; \Omega\\
\bs{\ell}  &= \epsilon^\partial( \bs{b}), \qquad \text{on} \; \partial \Omega.
\end{split}
\end{equation}
The assumption that $\mathcal{A}$ is positive definite is crucial. This implies that \eqref{eq:elliptic_system} is a {\em strongly elliptic system} of PDEs. It then follows, under appropriate smoothness conditions \cite{giaquinta_2012} on the latent boundary data $\epsilon^{\partial}(\bs{b})$ and the boundary geometry,  that  \eqref{eq:elliptic_system} has a unique solution $\bs{\ell} \in X_L$. We let $\mathcal{E}_{\mathcal{A}} : Y_L \rightarrow X_L$ denote the operator which maps boundary data to internal variables when solving the elliptic boundary value problem \eqref{eq:elliptic_system}.  The structure of the SINN mapping \eqref{eq:F_composition} is hence specifically designed to create a latent space in which passage of data from boundary to the interior is well-posed. This is achieved irrespective of the properties of the  PDE or the observation mapping structure \eqref{eq:PDE_abstract} from which the physical data was sampled.  

This third, and final, component of a SINN operator \eqref{eq:F_composition} is a {\em decoder}
\begin{equation} \label{eq:decoder_abstract}
\begin{split}
\delta &: X_L \longrightarrow X \\
& \quad \;\;\; \bs{\ell} \longmapsto \bs{u}
\end{split}
\end{equation}
which lifts a distribution of interior latent $\bs{\ell} \in X_L$ back into physical space $\bs{u} \in X$. Analogous to the boundary encoder, $\delta$ is assumed to be nonlinear and semi-local. That is, for any $\bs{x} \in \Omega$, the value of $(\delta \bs{\ell})(\bs{x})$ must only depend on the values of $\bs{\ell}$ in a small neighbourhood $N_{\bs{x}} \subset \Omega$ of $\bs{x}$. Again, this can be implemented using a single neural network $\mathcal{N} : \{ \bs{\ell}(\bs{y}) : \bs{y} \in N_{\bs{x}} \} \mapsto \bs{\ell}(\bs{x})$ which is applied repeatably to form the definition of the operator $\delta$, as described in detail in \S\ref{sec:decoder}.

In summary, a structure-informed neural network (SINN) $\mathcal{F}$ is an operator of the following form 
\[
\mathcal{F} = \left\{ \footnotesize \begin{array}{c} \text{semi-local} \\ \text{nonlinear NN} \\  \delta : X_L \rightarrow X \end{array} \right\} \circ \left\{ \footnotesize \begin{array}{c} \text{global} \\  \text{elliptic system} \\ \mathcal{E}_{\mathcal{A}}: Y_L \rightarrow X_L \end{array} \right\} \circ \left\{ \footnotesize \begin{array}{c} \text{semi-local} \\ \text{nonlinear NN} \\  \epsilon^\partial : Y \rightarrow Y_L \end{array} \right\} 
\]
The semi-local architecture of the encoder and decoder mappings is chosen specifically so as to restruct the number of degrees of freedom involved in defining the nonlinear components of the operator. The global transfer of information from boundary to interior is performed in the latent space via a well-posed elliptic system. This embeds a natural, yet very general, object into a SINN which is specifically tailored to the structure of the boundary observation problems that are our aim to solve. Furthermore, as will be described in \S\ref{sec:training}, a key advantage of using elliptic systems of PDEs is that identification of their coefficients can be performed in a computationally-efficient manner using only local training data. However, once trained, the resulting elliptic system can then be applied globally to give a SINN solution to the original boundary observation problem.    

In \S \ref{sec:generating_functions} we introduce the concept of a {\em generating function} which underpins the semi-local structure of the encoder and decoder operators, before introducing these operator formally and deriving their inherited mathematical properties. The method of training SINNs from data is described in \S \ref{sec:training} and its numerical implemention discussed in \S \ref{sec:Methodology}. Implementation of our approach on a pair of challenging test-cases is given in \S\ref{sec:num_egs}. Before this, we first comment briefly on the relation between the proposed SINN architecture and other, related, data-driven approaches to PDE analysis. 

\subsection{Relation of SINNs to existing methods} \label{sec:SINNS_vs_others}

The use of neural networks to solve PDEs has received much recent interest with the development of Physics-inspired Neural Networks (PINNs) \cite{raissi_2019}. In the context of solving a PDE of the form \eqref{eq:PDE_abstract}, the idea is to view the solution $\bs{u}$ as a mapping $\mathbb{R}^d \ni \bs{x} \mapsto \bs{u}(\bs{x}) \in \mathbb{R}^n$ and to therefore seek to construct a neural network $\mathcal{N}_P : \mathbb{R}^d \rightarrow \mathbb{R}^n$ which approximates the solution. The crucial step is to add so-called physics-inspired constraints, namely $\mathcal{L}(\mathcal{N}_P(\bs{x}),\bs{\lambda})_{|_\Omega}=0$ and $\left[\mathcal{B}(\mathcal{N}_P(\bs{x}))-\bs{b} \right]_{|_{\partial \Omega}}=0$, to force the constructed solution to satisfy the underlying PDE. 

In contrast to the SINN operators $\mathcal{F}:Y \rightarrow X$ which act between functions spaces, PINNs are finite-dimensional mappings that directly attempt to replicate the solution mapping $\bs{x} \mapsto \bs{u}(\bs{x})$. They require knowledge of the underlying PDE they seek to solve (i.e., of $\mathcal{L}, \bs{\lambda}$ and $\mathcal{B}$) and, when applied to boundary observation problems, must be trained using knowledge of the specific boundary data $\bs{b}$. In contrast, SINNs do not require such information: the identification of operators means that such boundary data is not required in the SINN methodology. 

The three-operator structure of the mapping $\mathcal{F} = \delta \circ \mathcal{E}_\mathcal{A} \circ \epsilon^\partial$ is widely used in a variety of data-driven approaches to low-order modelling. In Koopman-based modelling, for example, operators with this three-level structure are used to approximate the time-evolution of chaotic, infinite-dimensional, dynamical systems. In these approaches, the role of the central  operator $\mathcal{E}_\mathcal{A}$ is to model temporal evolution, rather than the passage of information from a domain's boundary to its interior as in this paper. The main distinction between the SINNs and the Koopman methodology is that, in the latter approach, the latent space is finite dimensional and the temporal operator a finite-dimensional ODE. 

This represents an important distinction with the SINN methodology. To explain, consider the case of the decoder $\delta$ operator, and assume that it maps from a finite dimensional latent space, say $\mathbb{R}^r$, into the infinite dimensional space of physical variables $X=L^2(\Omega,\mathbb{R}^n)$. The mismatch in dimensions between latent and physical space implies that the decoder must have an inherent method of translating finite-dimensional latent variables to infinite-dimensional functions. In Koopman-based approaches, this is typically achieved by considering a basis of functions $\{\Phi_i \}_{i=1}^N \subset X$ and letting $\delta$ involve a mapping from the latent space $\mathbb{R}^r$ to the coefficients $\{ \hat{f}_i\} \subset \mathbb{R}^N$ of a series expansion $\sum_{i=1}^N \hat{f}_i \Phi_i \in X$. A major challenge of this approach is the choice of an appropriate basis $\{ \Phi_i \}$ 
 and attempting to solve this problem has motivated a range of different Koopman-based methods \cite{wynn_2013,li_dmd_2017}. 
 
 In contrast, in the SINN approach developed here, the use of an elliptic system $\mathcal{E}_\mathcal{A}$ removes the need for assigning or identifying a set of basis functions, and therefore the imposition of unnecessary structure on the operator $\mathcal{F}$. Constructing an appropriate elliptic system only requires identifying the PDE coefficient matrix $\mathcal{A}$, which potentially offers a significant reduction in dimension to identifying a set of basis function $\{\Phi_i\} \subset X$. This advantage comes at the cost of requiring solution of a PDE, as opposed to an ODE, as the central component of the model $\mathcal{F}$. However, the SINN methodology deliberately imposes a well-posed elliptic structure which, in many cases, enables this PDE to be solved at accurately and at low cost using existing algorithms. In addition, as will be explained in \S\ref{sec:training}, since our aim is to identify a PDE, the cost function for SINN training can be chosen to involve only low-cost, local, solutions to elliptic systems during training. However, once trained, the identified elliptic systems can then be used to transfer information across a domain globally. 

Finally, the philosophy taken in this paper to identify {\em operators} using the SINN methodology is related to the recent interest in using neural networks to identify operators between function spaces, such as Neural Operators \cite{Kovachki_2023} or DeepONets \cite{lu_2021}. The Neural Operator framework \cite{Kovachki_2023} seeks to construct solution operators $G:\bs{\lambda} \mapsto \bs{u}$ which solve PDEs of the form \eqref{eq:PDE_abstract} with Dirichlet boundary conditions $\bs{b}=0$ using knowledge of their distributed parameters $\bs{\lambda}$. In this approach, $G$ transfers information globally in the domain $\Omega$ using an iterative sequence of integral operators whose kernels are identified using neural networks. For practicable computational implementation in model training, structure needs to be imposed in the integral kernels, such as using low-rank approximations, Convolutional Neural Networks, Graph Neural Networks \cite{pilva_2022_graph}, or Fourier Neural Operators \cite{li_2021_fourier}. Any such choice of structure is philosophically similar to need to prescribe a functional basis in the Koopman-based methodology described previously. Again,  the contrast to the SINN methodology is that by training an elliptic operators, only a relative small number of coefficients are required to enable global transfer of problem information, and this is achieved  without the need to impose any additional structure on the operator ansatz. We note, finally, that the DeepONet methodology \cite{lu_2021}, which can be viewed as a special case of the Neural Operator approach, also essentially requires the identification of a functional basis during training.

\section{Encoders and Decoders for SINNs} \label{sec:generating_functions}

In this section, we give a detailed description of the mathematical structure of the encoder and decoder operators required to create a SINN. We will describe three classes of operator:  interior encoders, boundary encoders, and decoders. As indicated in Figure \ref{fig:mapping_schematic}, only the boundary encoder and decoder are required to define a SINN mapping. However, as will be explained in \S \ref{sec:training}, interior encoders will be required to enable data-driven training. 

The semi-local structure of all encoder and decoder operators will be implemented by defining {\em generating functions} (GFs), which act as the building blocks of the SINN methodology. In each of the follow sections we first introduce a generating function, use it to define the respective operator, then comment on the regularity properties inherited by that operator. 

\subsection{Interior Encoders} 
 \label{sec:int_encoder}

 For the purposes of model training only, we will construct interior encoders $\epsilon$ which, given any distribution of physical variable $\bs{u} \in X$, transforms these into a distribution of latent variables $\bs{\ell} = \epsilon \bs{u}$ on the domain interior.\\ 

 \noindent
 {\em Interior Encoder GFs:} Given a compact set $0 \in E \subset \mathbb{R}^d$, a generating function for an interior encoder is any continuous, compact\footnote{A compact mapping is one which maps bounded subsets to relatively compact subsets.}, generally nonlinear mapping
\begin{equation} \label{eq:gf_int}
e : L^2(E,\mathbb{R}^n)  \longrightarrow \mathbb{R}^r.
\end{equation}
This should be thought of as a mapping 
\[
e : \left\{ \footnotesize \begin{array}{c} \text{Local patch of} \\ \text{ interior data} \end{array} \right\} \longmapsto \left\{ \footnotesize \begin{array}{c} \text{Latent} \\ \text{variables} \end{array} \right\}.
\]
which will be used to endow the interior encoder $\epsilon$ with the desired semi-local structure. \\

\noindent 
{\em Definition of Interior Encoders:} For any $\bs{x} \in \Omega$, define a local neighbourhood
\[
E_{\bs{x}} := \bs{x} + E =    \{ \bs{x} + \bs{y} : \bs{y} \in E\},
\]
and let $ \Omega_E:=\{ \bs{x} \in \Omega : E_{\bs{x}} \subset \Omega \}$ be the set of points whose neighbourhoods $E_{\bs{x}}$ are entirely contained in $\Omega$. These sets are shown in Figure \ref{fig:encoder}. 

Next let $\bs{u} \in X$. For any $\bs{x} \in \Omega_E$, a local function $\bs{u}_{\bs{x}} : E \rightarrow \mathbb{R}^n$ can be defined by
\[
\bs{u}_{\bs{x}}(\bs{y}) := \bs{u}(\bs{x}+\bs{y}), \qquad \bs{y} \in E.
\]
Given a generating function $e : L^2(E,\mathbb{R}^n) \rightarrow \mathbb{R}^r$, we then define an interior encoder by 
\begin{equation} \label{eq:interior_encoder}
\left(\epsilon \bs{u}\right)(\bs{x}):= e\left(\bs{u}_{\bs{x}} \right), \qquad \bs{x} \in \Omega_E,
\end{equation}

This definition should be thought of as mapping the physical data $\bs{u}$, viewed as a {\em function} in $X=L^2(\Omega,\mathbb{R}^n)$, to a new function $\epsilon \bs{u} : \Omega_E \rightarrow \mathbb{R}^r$. This allows the latent variables $\bs{\ell}(\bs{x})=(\epsilon \bs{u})(\bs{x})$ corresponding to $\bs{u}$ to be defined on the subdomain $\Omega_E$. 

It follows trivially from its definition that interior encoders are operators satisfying $\epsilon : L^2(\Omega, \mathbb{R}^n) \rightarrow L^2(\Omega_E,\mathbb{R}^r)$. However, the following result shows that the latent variable field created using the encoder $\epsilon$ are, in fact, continuous, uniformly bounded, functions.  

\begin{lemma} \label{lem:cts_int_encoder}
Let $e:L^2(E,\mathbb{R}^n) \rightarrow \mathbb{R}^r$ be an interior encoder generating function and let $\epsilon$ be defined by \eqref{eq:interior_encoder}. Then $\epsilon :  L^2(\Omega,\mathbb{R}^n) \rightarrow C(\Omega_E,\mathbb{R}^r)$.
\end{lemma}
\begin{proof}
See Appendix \ref{app:int_encoder}.
\end{proof}

\begin{figure}
    \centering
        \includegraphics[width=0.6\textwidth,clip=true,trim=2cm 1cm 8cm 1cm]{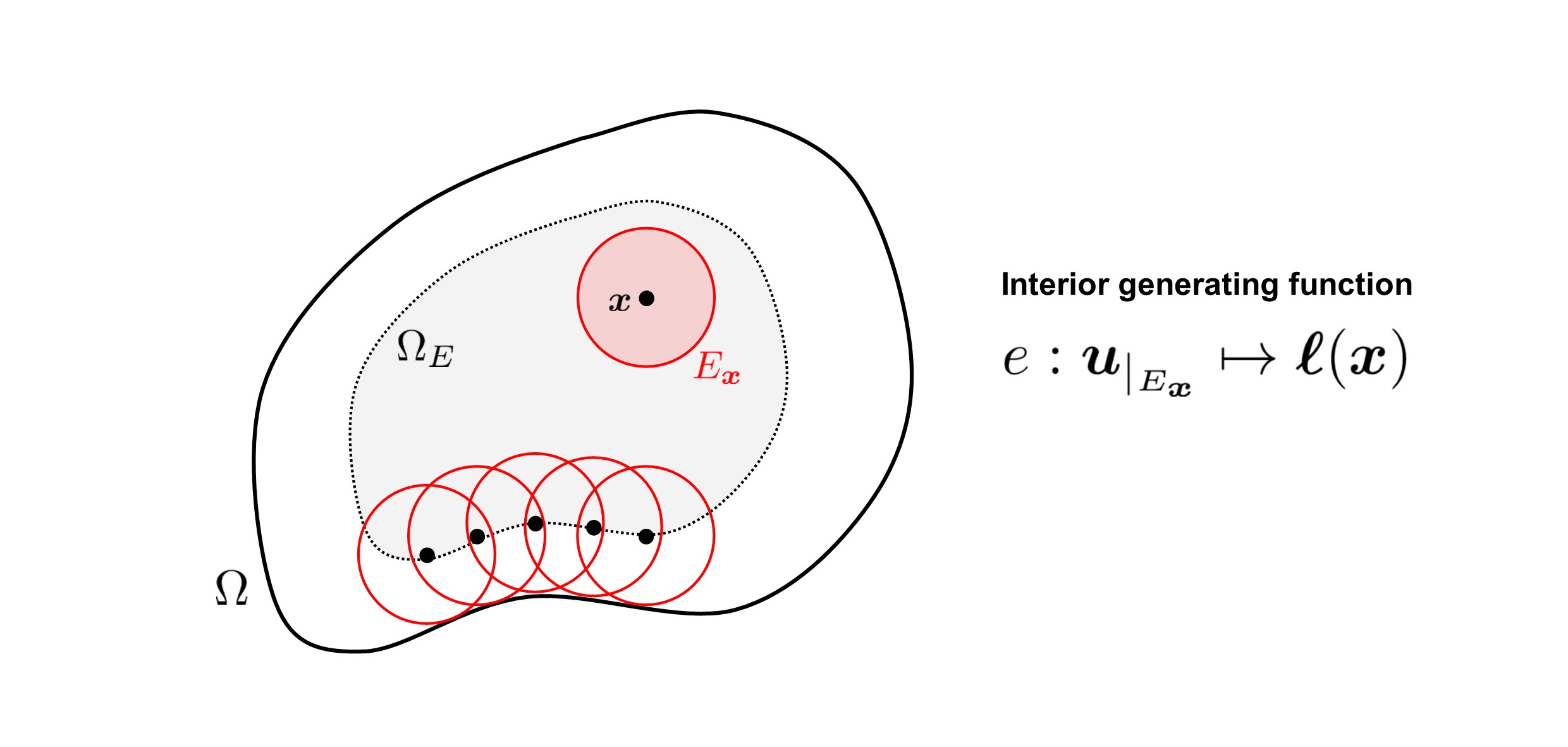}
    \caption{Schematic structure of action of an interior generating function $e:L^2(E) \rightarrow \mathbb{R}^r$. At a point $\bs{x} \in \Omega_E$, the local physical data $\bs{u}_{|_{E_x}}$ is viewed as a function in $L^2(E)$. The generating function then maps local information to the value of the latent varibles at $\bs{x}$ via  $e: \bs{u}_{|_{E_{\bs{x}}}} \mapsto \bs{\ell}(\bs{x}) \in \mathbb{R}^r$.} 
    \label{fig:interior_encoder}
\end{figure}

\subsection{Boundary Encoders} \label{sec:bdry_encoder}
We describe how to construct a boundary encoder $\epsilon^\partial$ which, given any distribution of boundary values $\bs{b} \in Y$, transforms these into a distribution of latent variables $\bs{\ell} = \epsilon^\partial \bs{b}$ on the boundary $\partial \Omega$. The construction is analogous to that of the interior encoder in \eqref{eq:interior_encoder} but with the added complication of including information about the boundary geometry. 

{\em Boundary encoder GFs:} Given a fixed, compact, set $0 \in E_\partial \subset \mathbb{R}^{d-1}$ containing the origin, a generating function for the boundary encoder is any continuous, compact, and generally nonlinear function
\begin{equation} \label{eq:gf_bdry_encoder}
e^\partial : L^2(E_\partial,\mathbb{R}^{n_\partial}) \times L^2(E_\partial,\mathbb{R}^d) \rightarrow \mathbb{R}^r.
\end{equation}
This should be understood as a mapping 
\[
e^\partial : \left\{ \footnotesize \begin{array}{c} \text{Section of} \\ \text{ boundary data} \end{array} \right\} \times \left\{ \footnotesize \begin{array}{c} \text{Section of} \\ \text{ boundary geometry} \end{array} \right\} \longmapsto \left\{ \footnotesize \begin{array}{c} \text{Latent} \\ \text{variables} \end{array} \right\}
\]
which will be used repeatably to define a semi-local boundary encoder operator. \\

\noindent
{\em Definition of Boundary Encoders:} We assume thoughout that $\partial \Omega$ is sufficiently regular that a normal vector $\bs{n}({\bs{z}}) \in \mathbb{R}^d$ and a tangent plane $T_{\bs{z}} \subset \mathbb{R}^{d-1}$ exists for every $\bs{z} \in \partial \Omega$. Each tangent plane $T_{\bs{z}}$ is defined in terms of a local-coordinate system with origin at $\bs{z}$ and whose basis vectors $(\bs{e}^{\bs{z}}_i)_{i=1}^{d-1}$ are orthogonal to $\bs{n}({\bs{z}})$. We  assume further that there exists a ball $B_R(\bs{z}) \subset \mathbb{R}^d$ of radius $R$ such that the local projection $P_{\bs{z}} : \partial \Omega \cap B_R(\bs{z})\rightarrow T_{\bs{z}}$ from the boundary to tangent plane $T_{\bs{z}}$ is one-to-one and, in addition, that there exists $\tau >0$, independent of  $\bs{z} \in \partial \Omega$, such that 
\begin{equation} \label{eq:tangent_assumption}
\{ t_i \bs{e}_i^{\bs{z}} : 0 \leq t_i < \tau\} \subset P_{\bs{z}}( B_R(\bs{z}) \cap \partial \Omega ) \subset T_{\bs{z}}, \qquad \bs{z} \in \partial \Omega
\end{equation}
and we also assume that 
\begin{equation} \label{eq:bdry_local_cond}
    E_{\partial} \subseteq (0,\tau)^{d-1}.
\end{equation}
Property \eqref{eq:bdry_local_cond} implies that we can view $E_\partial$ as a subset of the tangent plane, while \eqref{eq:tangent_assumption} then implies that a well-defined, continuous, inverse $P_{\bs{z}}^{-1} : \{t_i \bs{e}_i^{\bs{z}} : \bs{t} \in E_\partial \} \rightarrow \partial \Omega$ exists. A schematic illustration of this construction is shown in Figure \ref{fig:encoder}.       

This technical construction allows us, for each $\bs{z} \in \partial \Omega$, to define a function $\bs{b}_{\bs{z}} : E_\partial \rightarrow \mathbb{R}^{n_{\partial}}$, which depends on the boundary data local to $\bs{z}$, by
\begin{equation} \label{eq:boundary_data_inv_projection}
\bs{b}_{\bs{z}}(\bs{t}):=  \bs{b} \left(P^{-1}_{\bs{z}} (t_i\bs{e}_i^{\bs{z}}) \right), \qquad \bs{t} = (t_i)_{i=1}^{d-1} \in E_\partial,
\end{equation}
Similarly, we can also define a function $\bs{n}_{\bs{z}}:E_\partial \rightarrow \mathbb{R}^{d}$ which describes the boundary geometry local to $\bs{z}$ by
\begin{equation} \label{eq:boundary_normal_inv_proj}
\bs{n}_{\bs{z}}(\bs{t}):= \bs{n}\left(P^{-1}_{\bs{z}} (t_i\bs{e}_i^{\bs{z}}) \right), \qquad \bs{t} = (t_i)_{i=1}^{d-1} \in E_\partial.
\end{equation}

Next, using the boundary generating function $e^\partial$, the corresponding boundary encoder is defined by 
\begin{equation} \label{eq:bndry_encoder}
\left(\epsilon^{\partial} \bs{b}  \right)(\bs{z}) := e^\partial( \bs{b}_{\bs{z}}, \bs{n}_{\bs{z}}),  \qquad \bs{z} \in \partial \Omega. 
\end{equation}
Similar to the case of the interior encoder, since $e^\partial$ is assumed to be compact and continuous, an analogous proof to that of Lemma \ref{lem:cts_int_encoder} implies that  
\[
\epsilon^\partial : L^2(\partial \Omega,\mathbb{R}^{n_\partial}) \rightarrow C(\partial \Omega,\mathbb{R}^r).
\]
Hence,  the boundary latent variables $\bs{\ell}_{|_{\partial \Omega}} = \epsilon^\partial \bs{b}$ are continuous functions.

\begin{figure}
    \centering
        \includegraphics[width=0.6\textwidth,clip=true,trim=4cm 2cm 4cm 0.5cm]{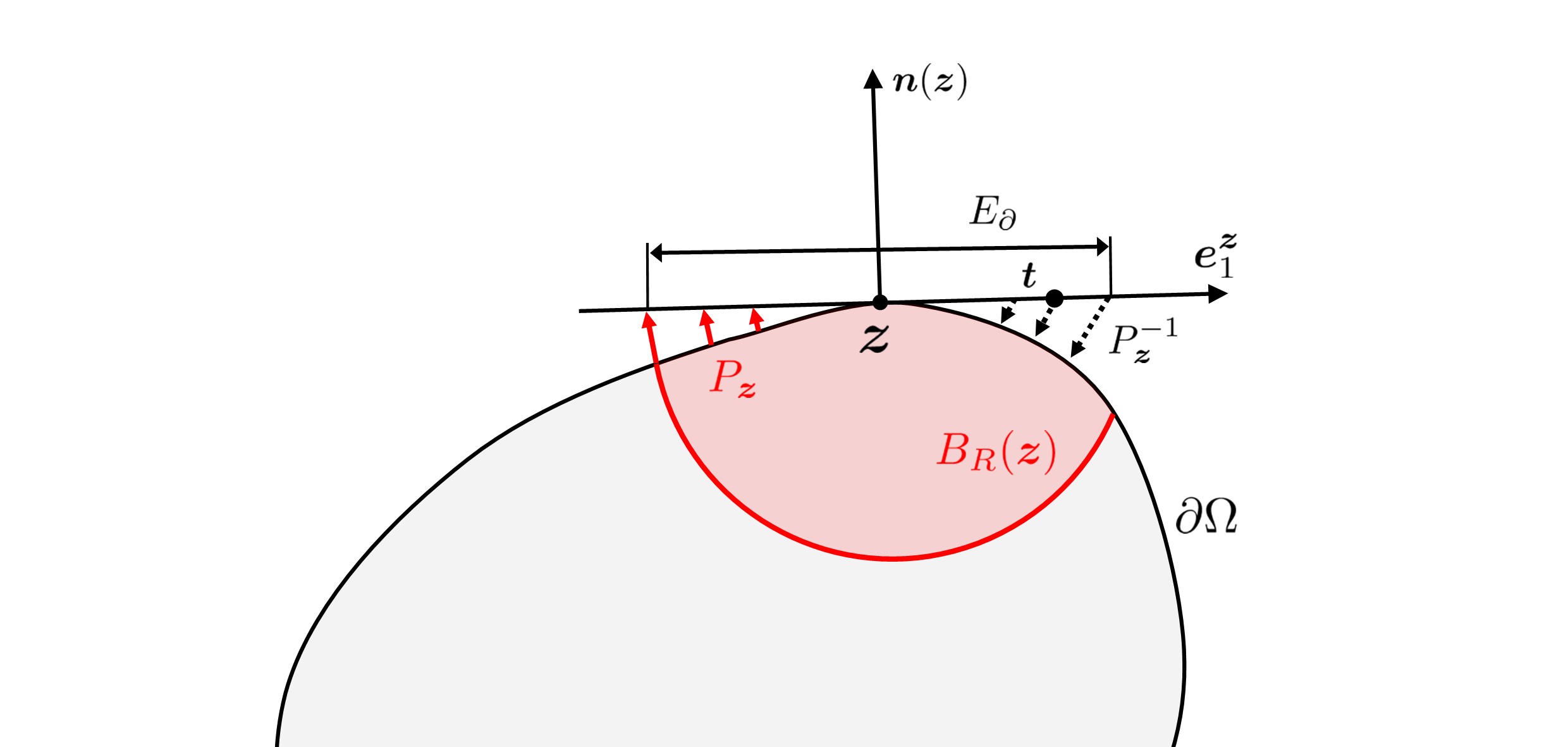}
        \caption{Schematic structure of the objects explaining the action of a boundary generating function $e^\partial : L^2(E_\partial) \times L^2(E_\partial) \rightarrow \mathbb{R}^r$. Local boundary data $\bs{b}_{\bs{z}}: \bs{t} \mapsto \bs{b}(P_{\bs{z}}^{-1} \bs{t})$ and geometry $\bs{n}_{\bs{z}} : \bs{t} \mapsto \bs{n}(P_{\bs{z}}^{-1} \bs{t})$ functions are considered as elements of $L^2(E_\partial)$. The generating function $e^\partial$ maps local boundary and geometric information to the value of the boundary latent variable at $\bs{z} \in \partial \Omega$ via  $e^\partial : (\bs{b}_{\bs{z}}, \bs{n}_{\bs{z}}) \mapsto \bs{\ell}(\bs{z})$.  } 
        \label{fig:encoder}
\end{figure}

\subsection{Decoders} \label{sec:decoder}

We describe how to construct a decoder mapping $\delta$ such that, given a distribution of latent variables $\ell \in X_L$, one can transform these into a distribution of physical variables $\bs{u} = \delta \bs{\ell}$ on the domain interior.

{\em Decoder GFs:} Given a compact, symmetric, set $0 \in D \subset \mathbb{R}^d$, a decoder generating function is any continuous, compact, and generally nonlinear mapping 
\begin{equation} \label{eq:gf_int_decoder}
d: \mathbb{R}^r \rightarrow C(D,\mathbb{R}^n).
\end{equation}
This GF should be thought of as follows: given latent variables $\bs{\ell}(\bs{y}) \in \mathbb{R}^r$ at a point $\bs{y} \in \Omega$, then $d(\bs{\ell}(\bs{y}))(\bs{x})$ gives a local prediction of the physical variables $\bs{u}(\bs{x}) \in \mathbb{R}^n$ for any $\bs{x} \in D_{\bs{y}} = \bs{y} + D$.  The idea is to use this map repeatably to build up a semi-local decoder operator. 
\vskip 1em

We define decoders in two situations, which we refer to as partition decoders and averaging decoders. 

{\em Definition of Partition Decoders:} In this case, it is assumed that there exist points $\{\bs{y}_i\}_{i=1}^{N_d} \subset \Omega$ such that the collections of sets $(\bs{y}_i + D)_{i=1}^{N_d}$ forms a disjoint partition of $\Omega$. Now, let $\bs{\ell} \in X_L$ be a latent variable distribution and let $\bs{x} \in \Omega$. Due to the assumed partition property, there is a unique index $j \in \{1,\dots,N_d\}$ such that $\bs{x} \in \bs{y}_j +D$. Consequently, $\bs{x}-\bs{y}_j \in D$ and we define a decoded value $\bs{u}_{\bs{\ell}}(\bs{x})$ by
\[
(\delta \bs{\ell})(\bs{x}) := d(\bs{\ell}(\bs{y}_j))(\bs{x}-\bs{y}_j). 
\]
Consequently, we can view $\delta$ as an operator $\delta : X_L \rightarrow X$ and we can create an approximation to the physical variables by letting $\bs{u}(\bs{x}) = (\delta \bs{\ell})(\bs{x})$. 

An advantage of using a partition decoder is that if $D$ is be chosen as a coarse discretization of $\Omega$, then this can reduce the computational cost of implementing the decoder. However, there are two potential disadvantages of this choice. First, requiring the decoder to extrapolate from latent to physical variables over a large set $D$ may introduce approximation errors to the solution. Second, while $\delta \bs{\ell}$ is guaranteed to be square-integrable (as an element of $X$), there is no guarantee that the resulting physical solution $\delta \bs{\ell}$ is smooth, or even continuous. If such a proprerty is desirable, the it is possible to instead implement the following notion of an averaging decoder.

{\em Definition of Averaging Decoders:} Let $\bs{\ell} \in X_L$ be a latent variable distribution and fix $\bs{x} \in \Omega$. Now, for any point $\bs{y}$ such that $\bs{x} \in D_{\bs{y}}$, it follows from the definition of decoder GFs that a prediction of the physical variables at $\bs{x}$ can be obtained using the function $d(\bs{\ell}({\bs{y}}))$. The idea is to average all such possible predictions. To simplify the resulting expression, note that since $D$ is symmetric, 
\[
\bs{x} \in D_{\bs{y}} \Leftrightarrow \bs{x}-\bs{y} \in D \Leftrightarrow \bs{y}-\bs{x} \in D \Leftrightarrow \bs{y} \in D_{\bs{x}},
\]
meaning that $\bs{x}$ can be predicted from any point $\bs{y} \in D_{\bs{x}} \cap \Omega$, as illustrated schematically in Figure \ref{fig:decoder}, and that the value of the prediction from the point $\bs{y}$ at $\bs{x}$ is $d(\bs{\ell}(\bs{y}))(\bs{x}-\bs{y})$.

\begin{figure}
\centering
\includegraphics[width=0.6\textwidth,clip=true,trim=3cm 0cm 6cm 0cm]{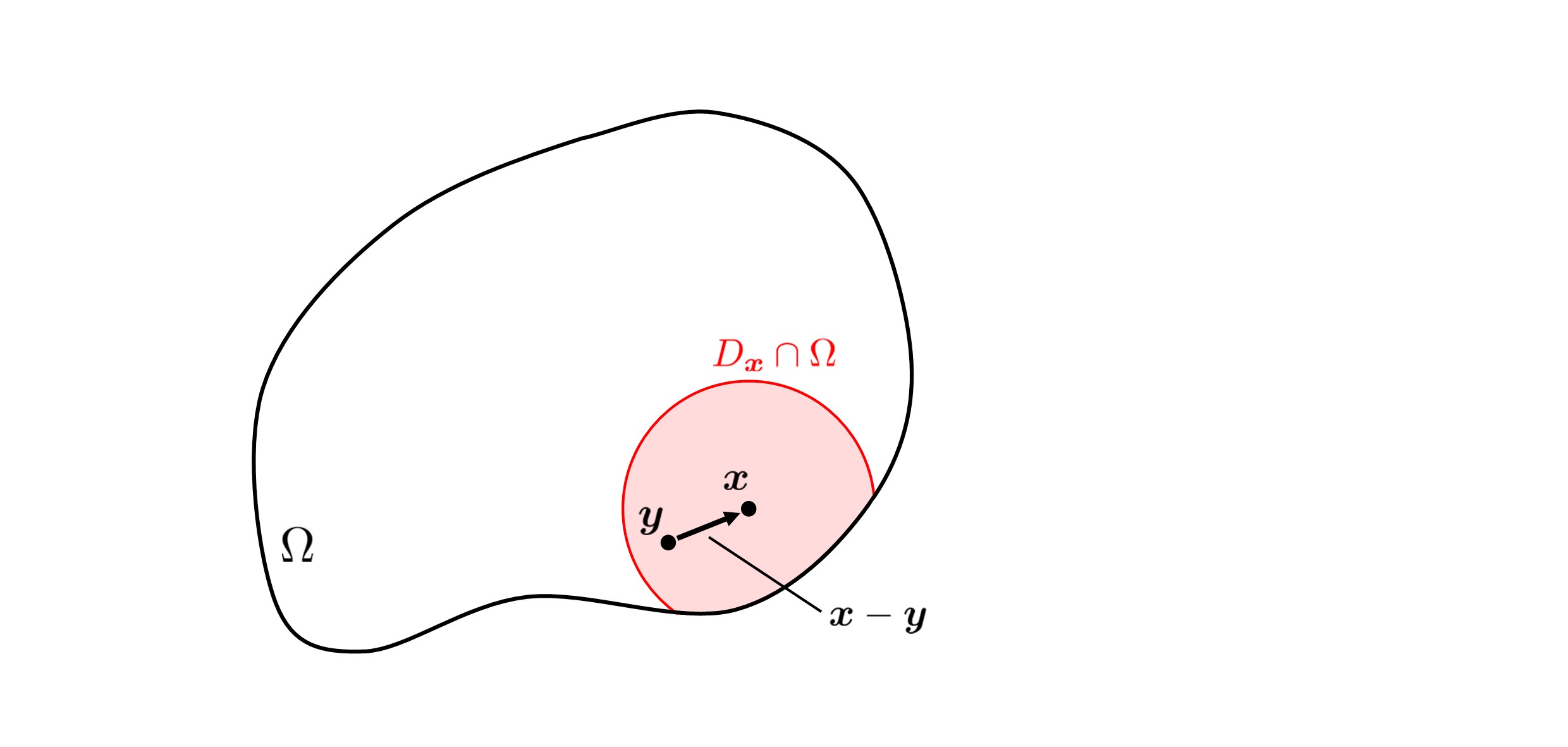}    
\caption{Schematic structure of the action of a decoder generating function $d:\mathbb{R}^r \rightarrow C(D)$. For any $\bs{x} \in \Omega$, the predicted physical value from a point $\bs{y} \in D_{\bs{x}} \cap \Omega$ is $d(\bs{\ell}(\bs{y}))(\bs{x}-\bs{y})$. The decoder operator $\delta$ averages over all such possible values.}
\label{fig:decoder}
\end{figure}

Consequently, given a function $\ell \in C(\Omega,\mathbb{R}^r)$, we define a decoder mapping $\delta$ by
\begin{equation} \label{eq:decoder}
(\delta \bs{\ell})(\bs{x}):= \frac{1}{|D_{\bs{x}} \cap \Omega|} \int_{D_{\bs{x}} \cap \Omega} d(\bs{\ell}(\bs{y}))(\bs{x}-\bs{y}) \, d\bs{y}, \qquad \bs{x} \in \Omega. 
\end{equation}

The following lemma shows that continuous generating functions create decoders which themselves produce continuous functions on the entire domain $\Omega$. 

\begin{lemma} \label{lem:decode_cts}
Let $d:\mathbb{R}^r \rightarrow C(D,\mathbb{R}^n)$ be a decoder generating function and let $\delta$ be defined by \eqref{eq:decoder}. Then $\delta : C(\Omega,\mathbb{R}^r) \rightarrow C(\Omega,\mathbb{R}^n)$. 
\end{lemma}
\begin{proof}
See Appendix \ref{sec:app_decode}.
\end{proof}

\subsection{Elliptic Systems} The final component required to define a SINN operator is an elliptic system. We simply refer to any symmetric, strictly positive definite, matrix
\begin{equation} \label{eq:gf_elliptic_system}
\mathcal{A} \in \mathbb{S}_{++}^{(rd)^2} 
\end{equation}
as a generating function from which an elliptic operator $D_\mathcal{A}$ and the associated boundary value problem \eqref{eq:elliptic_system} can be defined. This then generates the solution operator $\mathcal{E}_\mathcal{A} : Y_L \rightarrow X_L$.

\section{Training generating functions} \label{sec:training}

It is worth summarising the constructions developed in \S \ref{sec:generating_functions}. Given localisation sets $G:=(D,E,E_\partial)$ and local generating functions $(d,e,e^\partial,\mathcal{A})$, one can use \eqref{eq:interior_encoder}, \eqref{eq:bndry_encoder} and \eqref{eq:decoder} to define a globalisation mapping 
\[
\mathcal{G}_G :(d,e,e^\partial,\mathcal{A}) \mapsto (\delta,\epsilon,\epsilon^\partial,\mathcal{E}_\mathcal{A})
\]
which outputs an interior encoder $\epsilon$, boundary encoder $\epsilon^\partial$, decoder $\delta$ and elliptic system solution operator $\mathcal{E}_\mathcal{A}$. These components can then be combined to give a SINN operator  $\mathcal{F} = \delta \circ \mathcal{E}_\mathcal{A} \circ \epsilon^\partial$ in \eqref{eq:F_composition}. The aim now is to use the available data ensemble
\begin{equation} \label{eq:training_data}
\mathcal{U} = \left( \bs{u}_j(x), \bs{b}_j(\bs{z}) \right)_{j=1}^{N_T}, \qquad \bs{x} \in \Omega, \bs{z} \in \partial\Omega,
\end{equation}
to obtain identify an optimal generating functions and, consequently, optimal SINN operators.

\subsection{The cost function for SINN training}

Training will be posed as a minimisation problem, and a schematic for the cost function to be minimised is given by the four-stage process shown in Figure \ref{fig:training_run}. In the following, it is assumed that $(\bs{u},\bs{b}) \in \mathcal{U}$ is a snapshot selected from the training data ensemble.\\

\noindent
{\bf Stage 1:} Fix a set $\{\bs{p}_i\}_{i=1}^M \subset \Omega_E$ of {\em training points}. At each training point, it is assumed that a {\em training patch} exists, which is defined as the convex hull\footnote{The convex hull of a set of points is the smallest convex subset containing all such points.} of a set of points  $\{\bs{q}_{ij}\}_{j=1}^N \in \Omega \cup \partial \Omega$ local to $\bs{p}_i$, an example of which is shown in Figure \ref{fig:training_run}~(a). Specifically, for each $i$, we assume that there exist points  satisfying 
\begin{itemize}
  \item[$(i)$] $\bs{q}_{ij} \in \Omega_E \cup \partial \Omega$, for each $j=1,\dots,N$;
    \item[$(ii)$] $\bs{p}_i \in \text{int}(Q_i)$ where $Q_i=\text{conv}\{ \bs{q}_{ij}: j=1,\dots,N\}$;
    \item[$(iii)$]  $\bs{q}_{ij} \in \partial Q_i$, for each $j=1,\dots,N$.  
\end{itemize}
For any generating functions $e,e^\partial$, assumption (i) implies that we can compute the latent variables $\bs{\ell}_i:=e(\bs{u}_{\bs{p}_i})$ and 
\[
\bs{\ell}_{ij}:= \left\{ \begin{array}{rcl} e(\bs{u}_{\bs{q}_{ij}}), && \text{if} \; \bs{q}_{ij} \in \Omega; \\
e^\partial( \bs{b}_{\bs{q}_{ij}}, \bs{\eta}_{\bs{q}_{ij}}), && \text{if} \; \bs{q}_{ij} \in \partial \Omega.
\end{array} \right. , \qquad j=1,\dots,N. 
\]

\vskip 1em
\noindent
{\bf Stage 2:} Conditions $(ii)$ and $(iii)$ from Stage 1 imply that the convex hull $Q_i \subset \mathbb{R}^{d}$ is a polytope and that each  $\bs{q}_{ij}$ is an exterior point on its boundary $\partial Q_i \subset \mathbb{R}^{d-1}$. Linear interpolation can then be used to obtain a function $\bs{f}_i \in C(\partial Q_i,\mathbb{R}^r)$ satisfying $\bs{f}_i(\bs{q}_{ij}) = \bs{\ell}_{ij}$ for each $j=1,\dots,N$. This process is indicted in Figure \ref{fig:training_run}~(b).

\begin{figure}
\centering
\subfigure[Training Patch]{
\includegraphics[width=0.4\textwidth,clip=true,trim=0.5cm 2cm 12cm 2cm]{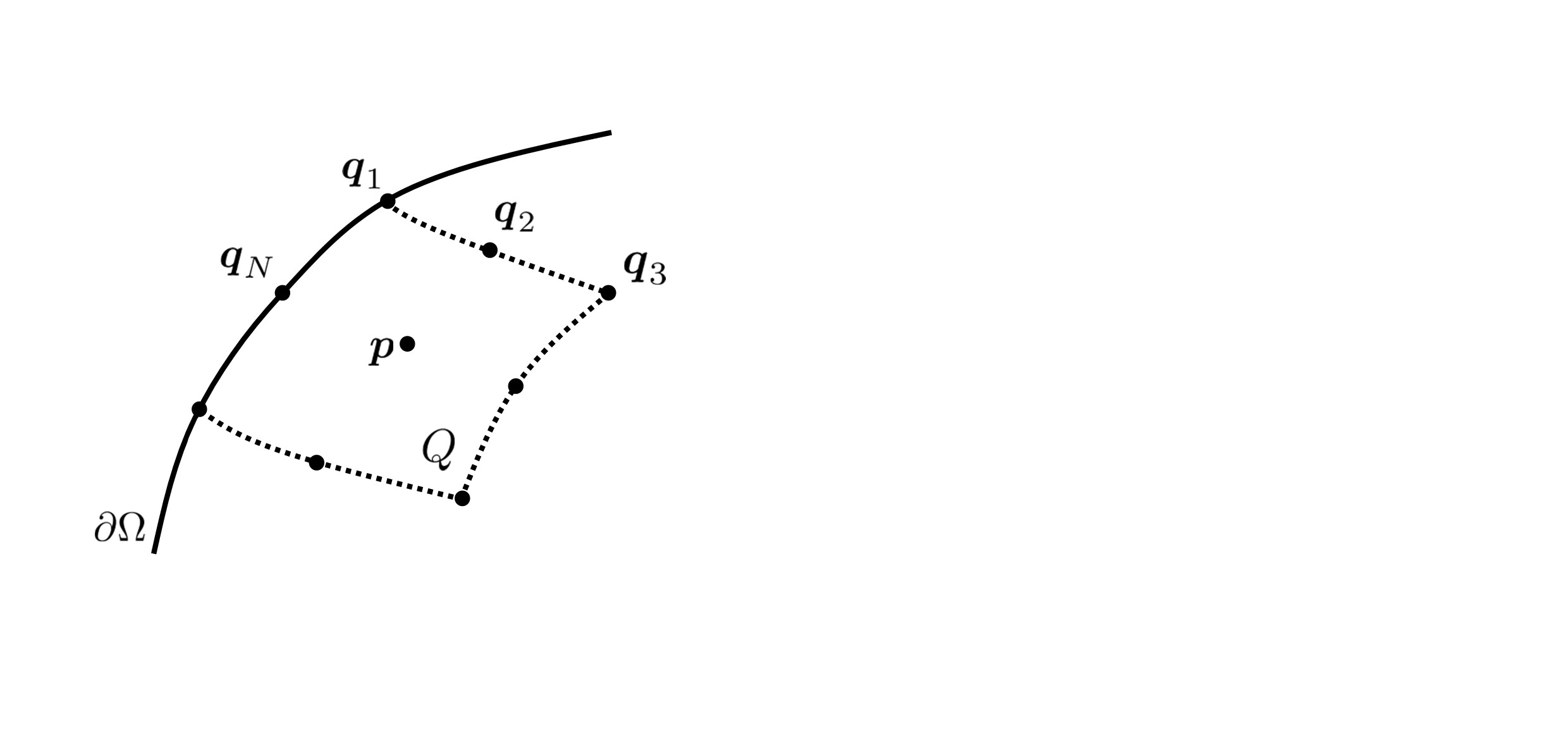}   
}
\qquad 
\subfigure[Latent variables on $\partial Q$]{
\includegraphics[width=0.4\textwidth,clip=true,trim=0.5cm 2cm 12cm 2cm]{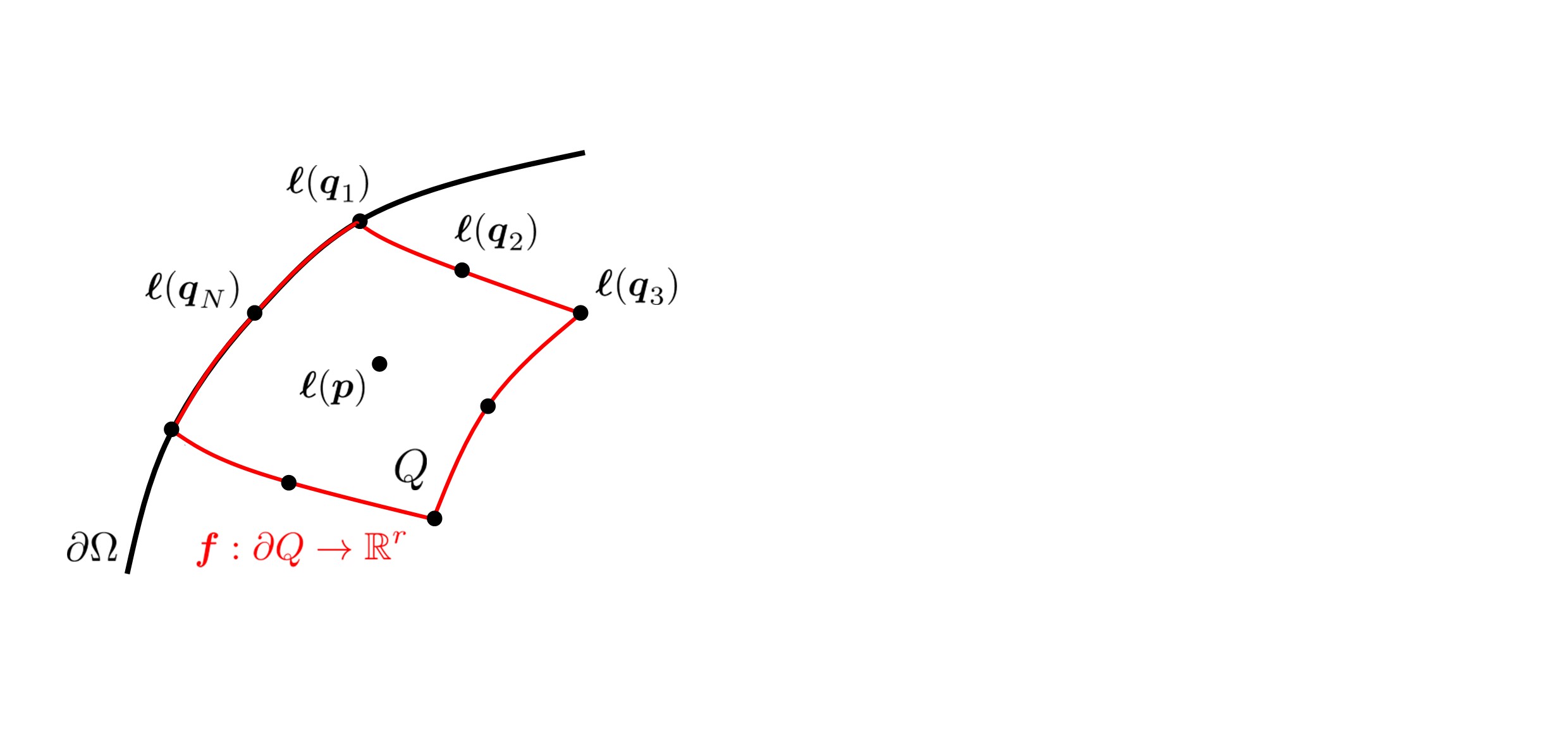}   
}
\\
\subfigure[Elliptic extension of latent variables to $Q$.]{
\includegraphics[width=0.4\textwidth,clip=true,trim=0.5cm 2cm 12cm 2cm]{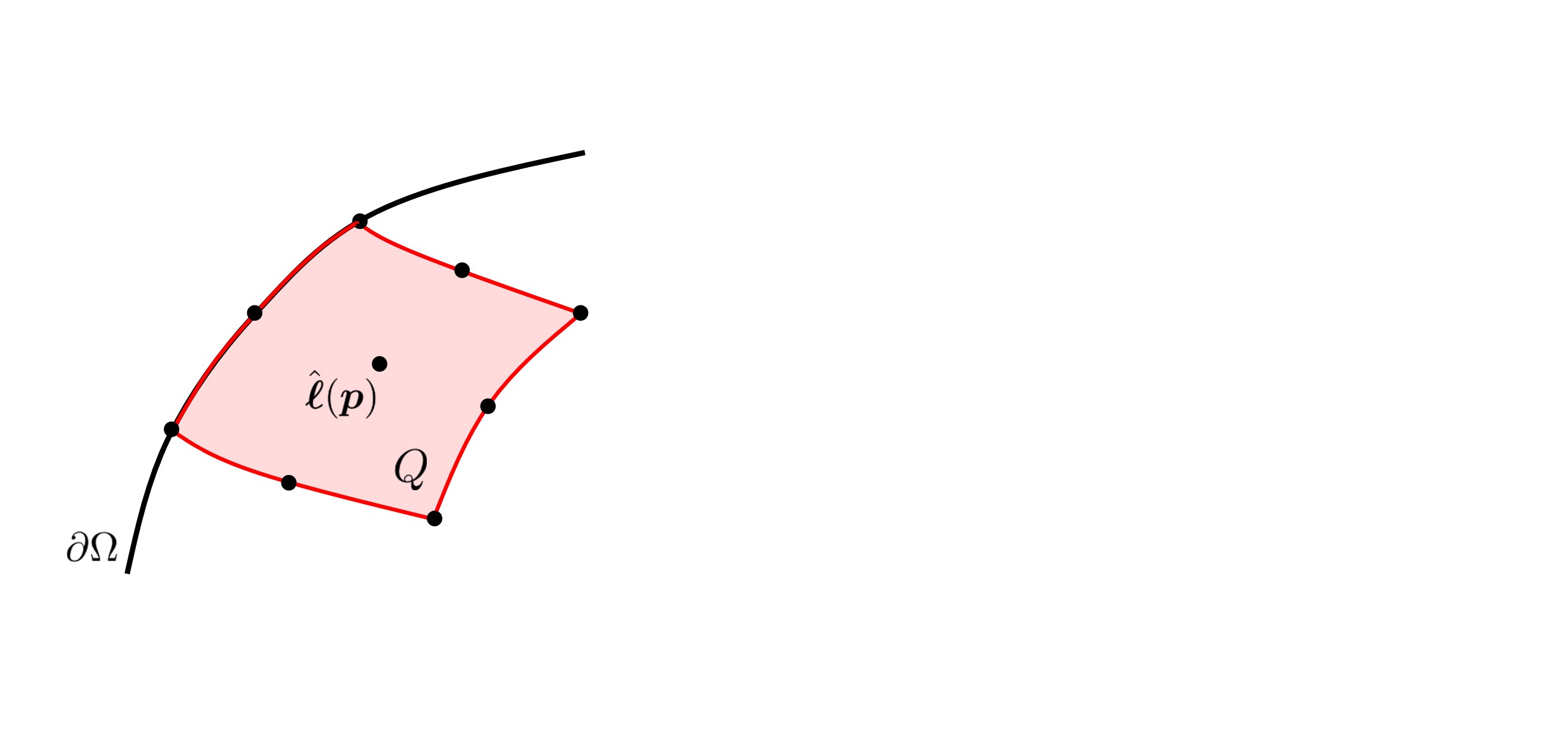}   
}
\qquad
\subfigure[Decoded physical variables at $\bs{p}$.]{
\includegraphics[width=0.4\textwidth,clip=true,trim=0.5cm 2cm 12cm 2cm]{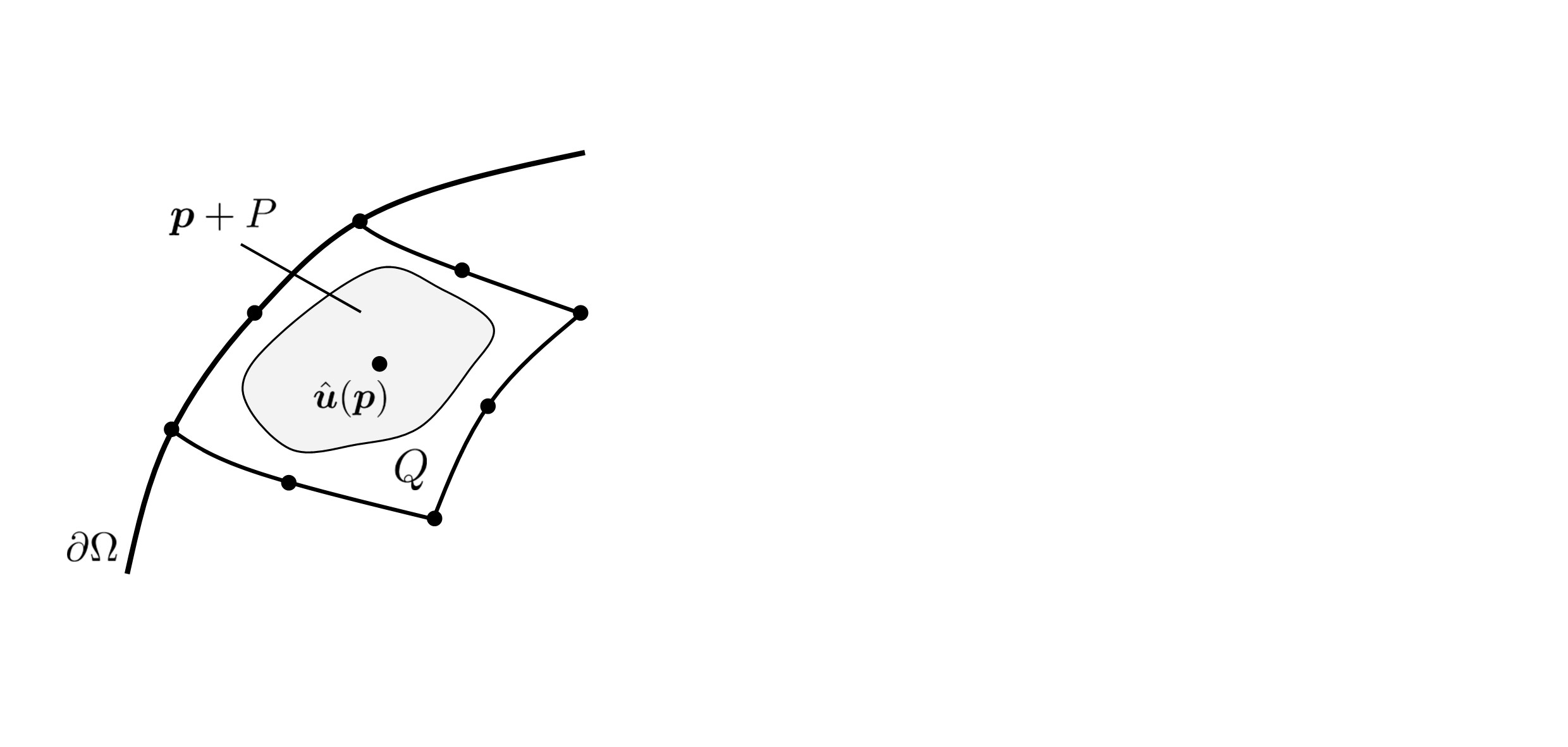}   
}
\caption{Schematic overview of variables involved in a training run. (a) A training patch with $Q$, containing central point $\bs{p}$, formed as the convex hull of exterior points $\{ \bs{q}_i\}$. (b) Latent variables computed to form a function $\bs{f}:\partial Q \rightarrow \mathbb{R}^r$ on the training patch boundary. (c) Elliptic extension defines latent variables in $Q$, in particular at $\hat{\ell}(\bs{p})$. (d) The decoder allows a prediction of the physical variables at $\bs{p}$.}
\label{fig:training_run}
\end{figure}

\vskip 1em 
\noindent
{\bf Stage 3:} Given a generating matrix $\mathcal{A} \in \mathbb{R}^{dr \times dr}$, define the associated linear elliptic operator $D_\mathcal{A}$, and solve the boundary value problem
\begin{equation} \label{eq:patch_elliptic_system}
\begin{split}
    D_\mathcal{A} \hat{\bs{\ell}} &= 0 \qquad \text{in} \; Q_i, \\ \hat{\bs{\ell}}_{|_{\partial Q_i}} &= \bs{f}_i \qquad \text{on} \;\partial Q_i,
\end{split}
\end{equation}
on the training patch $Q_i$. From this, a predicted value $\hat{\bs{\ell}}_i := \hat{\bs{\ell}}(\bs{p}_i) \in \mathbb{R}^r$ can be obtained, as shown in Figure \ref{fig:training_run}~(c). It is then natural to define the error function 
\[
\Psi_1((\bs{u},\bs{b}), (e,e^\partial,\mathcal{A})):=\frac{1}{Mr} \sum_{i=1}^M \left\| \bs{\ell}_i - \hat{\bs{\ell}}_i \right\|_2^2.
\]
which quantifies the error, averaged over all training points $\{ \bs{p}_i \}_{i=1}^M$, between the encoded latent variables computed using the generating function $e$, and their predictions from the boundary $\partial Q_i$ using the elliptic system $\mathcal{E}_\mathcal{A}$. \\

\noindent
{\bf Stage 4.} Given a decoder $d:\mathbb{R}^r \rightarrow C(D,\mathbb{R}^n)$, we create predictions for the physical variables, in the local sets $D_{\bs{p}_i}$ of points close to the training points, using 
\[
\hat{\bs{u}}(\bs{x}):= d(\hat{\bs{\ell}}_i)(\bs{x} - \bs{p}_i), \qquad \bs{x} \in D_{\bs{p}_i}. 
\]
A second error function 
\[
\Psi_2((\bs{u}, \bs{b}),(d,e,e^\partial,\mathcal{A})):= \frac{1}{M|D|} \sum_{i=1}^M \int_D | \bs{u}(\bs{p}_i + \bs{y}) - \hat{\bs{u}}(\bs{p}_i + \bs{y})|_2^2 d\bs{y}.
\]
then quantifies whether the decoder generating function $d$ is able to accurately recreate the physical data, averaged across a subset of $\Omega$ local to the chosen training points.

The four-stage process described above allows us, for each data point $(\bs{u},\bs{b}) \in \mathcal{U}$ to define two functions $\Psi_1,\Psi_2$ which quantify the error associated with a given quadruple of generating functions $\mathcal{X}:=(d,e,e^\partial,\mathcal{A})$. We  therefore define the ensemble cost function by
\begin{equation} \label{eq:cost_function}
\Psi\left( \mathcal{U}, \mathcal{X}\right):= \frac{1}{N_T} \sum_{j=1}^{N_T} \left[ \Psi_1((\bs{u}_j,\bs{b}_j), \mathcal{X}) + \alpha \Psi_2((\bs{u}_j,\bs{b}_j), \mathcal{X}) \right],
\end{equation}
where $\alpha>0$ is a weighting parameter.

\section{Numerical Implementation} \label{sec:Methodology}

 In  \S \ref{sec:generating_functions} it was shown how the generating functions $(e,e^\partial, d, \mathcal{A})$ act as building blocks for SINNs. The generating functions $(e,e^\partial,d)$ and, however, still {\em infinite dimensional} nonlinear functionals. A feasible approach to SINN training must therefore parameterise the infinite-dimensional generating functions $(e,e^\partial,d)$ by mappings between finite-dimensional spaces. The assumed semi-local structure of these operators will be used to achieve this in a simple manner, avoiding the need to impose any restrictive structure on the resulting approximation.  

\subsection{Interior generating functions} \label{sec:num_interior_gf}

The aim is to create an ensemble of possible mappings  $e:L^2(E,\mathbb{R}^n)\rightarrow \mathbb{R}^r$ using only finite dimensional functionals. The first step is to partition $E \subset \mathbb{R}^d$ into a union of $N_E$ subsets,
\begin{equation} \label{eq:E_partition}
E = \bigcup_{i=1}^{N_E} E_i. 
\end{equation}
Then, given any function  $\bs{u} \in L^2(E,\mathbb{R}^n)$, we can create a vector $\mathcal{D}\bs{u} \in \mathbb{R}^{n \times N_E}$ by taking its the average of $\bs{u}$ over each of the $N_E$  subsets. Consequently, any {\em finite-dimensional} mapping $\tilde{e} : \mathbb{R}^{n \times N_E} \rightarrow \mathbb{R}^r$ can be used to define an encoder $e:L^2(E,\mathbb{R}^n) \rightarrow \mathbb{R}^r$ by forming the composition $e = \tilde{e} \circ \mathcal{D}$.

\subsection{Boundary generating functions} \label{sec:num_bdry_gf}

The aim is to create boundary generating function $e^\partial : L^2(E_\partial,\mathbb{R}^{n_\partial}) \times L^2(E_\partial,\mathbb{R}^{d}) \rightarrow \mathbb{R}^r$. Using the same idea as before,  partition  $E_\partial \subset \mathbb{R}^{d-1}$ into $N_{E_\partial}$ sets
\begin{equation} \label{eq:E_bdry_partition}
E_\partial = \bigcup_{i=1}^{N_{E_\partial}} (E_\partial)_i. 
\end{equation}
Then, for any functions $\bs{f} \in L^2(E_\partial,\mathbb{R}^{n_\partial})$ and $\bs{\eta} \in L^2(E_\partial,\mathbb{R}^{d})$ we can form vectors $\mathcal{D}^\partial \bs{f} \in \mathbb{R}^{n_\partial \times N_{E_\partial}}$ and $\mathcal{D}^{\partial} \bs{\eta} \in \mathbb{R}^{d \times N_{E_\partial}}$ by averaging these functions over each subset $(E_\partial)_i$. Consequently, any finite-dimensional map $\tilde{e}^\partial : \mathbb{R}^{(n_\partial + d)\times N_{E_\partial}} \rightarrow \mathbb{R}^r$ induces a boundary encoder $e^\partial :L^2(E_\partial,\mathbb{R}^{n_\partial}) \times L^2(E_\partial,\mathbb{R}^{d}) \rightarrow \mathbb{R}^r$ via the composition $e^\partial := \tilde{e}^\partial \circ \mathcal{D}^\partial$. 

A summary of the constructions developed so-far is given in Table \ref{table:encoders}. It should be emphasised that due to the {\em semi-local} role of the encoder generating function $e$, it is not important to impose any particular structure on the partitions \eqref{eq:E_partition} or \eqref{eq:E_bdry_partition}. The only requirement is to form a sufficiently resolved local  approximation to the underlying data.

\begin{table}[h!]
\centering
\footnotesize
\begin{tabular}{|c |c| c|} 
 \hline
  & Interior Encoder & Boundary Encoder \\ [0.5ex] 
 \hline\hline
Semi-local &  $\begin{array}{rcl} \epsilon :L^2(\Omega) &\rightarrow & C(\Omega) \\  \bs{u} &\stackrel{\eqref{eq:interior_encoder}}{\mapsto} & 
e(\bs{u}_{\bs{x}}) \end{array}$ & $\begin{array}{rcl} \epsilon^\partial : L^2(\partial \Omega) \times L^2(\partial \Omega) &\rightarrow & C(\partial \Omega) \\ (\bs{b}, \bs{n}) &\stackrel{\eqref{eq:bndry_encoder}}{\mapsto} & e^\partial(\bs{b}_{\bs{z}},\bs{n}_{\bs{z}}) \end{array}$ \\\hline 
GF  &  $\begin{array}{c} e :L^2(E) \rightarrow  \mathbb{R}^r \\ e = \tilde{e} \circ \mathcal{D} \end{array}$&   $\begin{array}{c} e^\partial :L^2(E_\partial) \times L^2(E_\partial) \rightarrow  \mathbb{R}^r \\ e^\partial = \tilde{e}^\partial \circ \mathcal{D}^\partial \end{array}$ \\ \hline
$\begin{array}{r} \text{Finite} \\ \text{Dim.} \end{array}$   &  $\tilde{e} : \mathbb{R}^{n \times N_E} \rightarrow \mathbb{R}^r$  & $\tilde{e}^\partial : \mathbb{R}^{(n_\partial + d) \times N_{E_{\partial}}} \rightarrow \mathbb{R}^r $  \\ \hline
\end{tabular}
\caption{Structure of the Interior and Boundary Encoders.}
\label{table:encoders}
\end{table}

\subsection{Decoder Generating Functions} \label{sec:num_decoder_gf}

We aim to create functions $d:\mathbb{R}^r \rightarrow C(D,\mathbb{R}^n)$. To do this, choose $N_D$ points $\{d_i\}_{i=1}^{N_D} \subset \mathbb{R}^d$ whose convex hull contains $D \in \mathbb{R}^d$, and let $\mathcal{I} : \mathbb{R}^{N_D} \rightarrow C(D,\mathbb{R}^n)$ be any interpolation operator, such as linear interpolation, which continuously extends known functional values at the points $d_i$ to the whole of $D$. Then, any mapping $\tilde d:\mathbb{R}^r \rightarrow \mathbb{R}^{N_D}$ induces a decoder generating function via $d := \mathcal{I} \circ \tilde d$.

\subsection{Elliptic system solution on training patches} \label{sec:num_train_patch} 

Effective training of interior and boundary encoders requires a choice of training patches which provide a good sample of both the domain interior and the domain boundary. For simplicity, we only describe the case of rectangular domains $\Omega$. In this case, by considering training patches $Q = \text{conv}\{ \bs{q}_i\} \subset \bar{\Omega}$ whose boundary $\partial Q$ is also rectangular, these may be chosen to either lie entirely in the domain interior, or to  coincide with a portion of the boundary $\partial \Omega$. 

After appropriate interpolation of the latent variables $\ell(\bs{q}_i)$ to the rectangular boundary $\partial Q$, standard second-order finite difference schemes can be used to solve the required elliptic system \eqref{eq:patch_elliptic_system} on each training patch to obtain the prediction $\hat{\bs{\ell}}(\bs{p})$ at the training point $\bs{p} \in Q$.

\subsection{Cost function minimisation} \label{sec:cost_min}

Suppose that an ensemble $\mathcal{U}$ of training data of the form \eqref{eq:training_data} is available. Suppose that the latent space dimension $r$, the partition dimensions $N_E, N_{E_\partial}$ for the encoders, and the decoder discretisation dimension $N_D$  are all given. The problem of model training is now reduced to finding mappings $\tilde{e}:\mathbb{R}^{n \times N_E} \rightarrow \mathbb{R}^r, \tilde{e}^\partial : \mathbb{R}^{(n_\partial + d)\times N_{E_{\partial}}} \rightarrow \mathbb{R}^r, d : \mathbb{R}^{r} \rightarrow \mathbb{R}^{N_D}$ and a matrix $\mathcal{A} \in \mathbb{R}^{rd \times rd}$ which minimise that for which the cost function \eqref{eq:cost_function} is minimised.

Each of the generating functions $\tilde{e},\tilde{e}^\partial$ and $d$ are assumed to be fully connected feed-forward neural networks. A neural network (NN) with $L$ layers is a mapping $\mathcal{N}: \mathbb{R}^{m_{I}} \rightarrow \mathbb{R}^{m_O}$ formed by repeated composition of a prescribed nonlinear {\em activation function} $\phi:\mathbb{R} \rightarrow \mathbb{R}$ and a series of affine maps $A_i : \bs{x} \mapsto W_i \bs{x} + \bs{b}_i$, for $i=1,\dots,L$. 
Here, $W_i \in \mathbb{R}^{r_i \times r_{i-1}}$ and $\bs{b}_i \in \mathbb{R}^{r_i}$ are the free parameters of the neural network, $r_i$ are the number of neurons in the $i^\text{th}$ layer, and $r_0 = m_I, r_L= m_O$.  The output of the $i^\text{th}$ layer of the network is given by $\mathcal{N}_i(\bs{x}):= W_i \phi(\mathcal{N}_{i-1}(\bs{x})) + \bs{b}_i$, $i \geq 2$, where $\phi$ acts component-wise. For an input $\bs{x} \in \mathbb{R}^{m_I}$, and letting $\mathcal{N}_1\bs{x} = W_1 \bs{x} + \bs{b}_1$, the output of the neural network, after iterating its $L$ layers,  is  $\mathcal{N}(\bs{x}) = \mathcal{N}_L(\bs{x})$. In this paper, we fix 
\[
\phi(x) = \text{ReLU}(x) = \left\{ \begin{array}{rcl} x, &&  x \geq 0, \\ 0 && x < 0 ,\end{array} \right.
\]
meaning that the tunable parameters of the considered neural networks are  $\Theta = (W_i,\bs{b}_i)_{i=1}^L$. We then write $\mathcal{N} = \mathcal{N}_\Theta$ to emphasise this dependency.

The aim of model training is to find a quadruple $\mathcal{X}:=(e,e^\partial,d,\mathcal{A})$ which minimises the modelling residual $\Psi(\mathcal{U},\mathcal{X})$ defined in \eqref{eq:cost_function}. To impose the positive definiteness constraint \eqref{eq:patch_elliptic_system} on $\mathcal{A}$, we use a logarithmic barrier function and instead solve the optimisation problem
\begin{equation} \label{eq:numerical_optimisation}
\begin{split}
\min_{\Theta,\mathcal{A}} \quad & \Psi(\mathcal{U},\mathcal{X}) - \rho \log{(\det{\mathcal{A}})}\\
                   \quad & \Theta = (\Theta_e,\Theta_{e^\partial}, \Theta_d), \\
                   \quad & \mathcal{A} \in \mathbb{R}^{dr \times dr}, \\
                   \quad & e = \tilde{e} \circ \mathcal{D}, \; e^\partial := \tilde{e}^\partial \circ \mathcal{D}^\partial, \; d := \mathcal{I} \circ \tilde d,\\
                   \; & \tilde{e} = \mathcal{N}_{\Theta_e}, \; \tilde{e}^\partial = \mathcal{N}_{\Theta_{e^\partial}}, \; d = \mathcal{N}_{\Theta_d}\\
                   \quad & \rho >0, \alpha >0. 
\end{split}
\end{equation}
The input and output dimensions of each neural network in \eqref{eq:numerical_optimisation} are prescribed by the encoder and decoder structure described in \S \ref{sec:num_interior_gf}, \S\ref{sec:num_bdry_gf}, and \S\ref{sec:num_decoder_gf}. The number of layers and neurons in each neural network is problem dependent and will be specified in \S\ref{sec:num_egs} for the particular numerical examples considered. The $-\log{(\det{\mathcal{A}})}$ term provides a {\em barrier} in the sense that $-\log{(\det{\mathcal{A}})} \rightarrow \infty$ as $\det{A} \rightarrow 0$, which penalises sign-changes of the eigenvalues of $\mathcal{A}$, and hence promotes  positive definiteness.

The optimisation problem \eqref{eq:numerical_optimisation} is solved using a stochastic gradient descent approach. At each iteration, a random set of training patches of the form described in \S \ref{sec:num_train_patch} are created. The gradient of cost function \eqref{eq:cost_function}, dependent upon the chosen training patches, is then computed using automatic differentiation an appropriate step in the decision variables is taken. The Adam
gradient-based optimization algorithm implemented in the \texttt{TensorFlow} package is used to identify a local minimum.
Since solutions to the elliptic system \eqref{eq:elliptic_system} are invariant upon rescaling of $\mathcal{A}$, after each iteration the elliptic decision variables are updated via $\mathcal{A} \mapsto \mathcal{A}/ (\text{det}(\mathcal{A}))^{(1/(rd)^2)}$ to maintain the value of the determinant to be unity. The process of random training patch selection and gradient-based weight updates is then iterated until the cost $\Psi$ has converged to a local minimum.

\section{Numerical Examples} \label{sec:num_egs}

We consider two nonlinear PDEs to demonstrate performance SINNs for solving boundary observation problems. 

\subsection{A nonlinear heat equation} \label{subsec:nl_heat}

Let $\Omega = [0,1]\times [0,1] \subset \mathbb{R}^2$ and suppose that $u(x,y)$ satisfies the PDE
\begin{equation} \label{eq:heat_nl}
\begin{split}
\nabla \cdot(e^{u} \nabla u) &= 0,\qquad  \text{in} \; \Omega,\\
u &= g, \qquad \text{on} \; \partial \Omega. 
\end{split}
\end{equation}
This can be viewed as a the steady solution of a nonlinear diffusion equation for which the diffusivity, $e^{u(x,y)}$, depends on  the local solution $u(x,y)$. An alternative view is that \eqref{eq:heat_nl} is equivalent to the nonlinear PDE $\Delta u =-|\nabla u|^2$. The motivation for studying this example is that the linearised PDE is simply the Laplace equation $\Delta u=0$, meaning that this example will facilitate a careful comparison of the SINN methodology to more traditional approaches which directly employ a system's linearised dynamics with an encoder/decoder architecture. 

Data for training and testing is obtained by first creating $10^3$ boundary functions $g_i \in L^2(\partial \Omega,\mathbb{R})$, with $N_T=900$ of these used for training and the remaining $100$ boundary functions used for testing.  The boundary functions $g_i$ are created as random sums of sinusoids. To describe this process, for any boundary point $\bs{z}=\bs{z}(x,y) \in  \partial \Omega$, let $\alpha(\bs{z})$ be the angle, measured anticlockwise, using a co-ordinate system with origin at the centre of the square domain $\Omega$, namely
\[
\alpha(\bs{z}(x,y)) = \text{atan2}\left(x-0.5, y-0.5 \right), \qquad \bs{z} \in \partial \Omega.
\]
By sampling coefficients $X_i,Y_i \sim N(0,1)$ from standard Normal distributions, we first let 
\[
\tilde{g}_i(\bs{z}) = \sum_{n=1}^{4} \frac{X_{n} \sin(n \alpha(\bs{z})) + Y_{n} \cos(n \alpha(\bs{z})) }{n}, \qquad \bs{z} \in \partial \Omega. 
\]
The cosine component creates a random phase shift, while higher-order sinusoids are moderately attenuated to encourage the lower frequency data. Each sampled boundary function is then normalized to define the final data boundary function $g_i = \tilde{g}_i / (\Delta_g)$ where $\Delta_g$ is randomly sampled from a triangular distribution with pdf 
\[
f(\Delta_g) = 
\begin{cases} 
\Delta_g /8 , & \text{if } 0 \leq \Delta_g \leq 4, \\
0, & \text{otherwise}.
\end{cases}
\]
This involved approach create an ensemble of boundary functions which have mostly large differences between their largest and smallest value.

\subsubsection{Numerical solution and SINN implementation} For each boundary data function $g_i \in L^2(\partial \Omega,\mathbb{R})$, the PDE \eqref{eq:heat_nl} is solved on a uniform grid $38 \times 38$ grid with a Newton Linearization Method with finite difference equations to obtain solution data $u_i \in L^2(\Omega,\mathbb{R})$. 
%
Given the computational domain discretisation, we view  $\Omega$ as the union of $38 \times 38$ square elements $\mathcal{T} = [1/38] \times [1/38]$, with any solution $u(x,y)$ to \eqref{eq:heat_nl} assumed to have a single value in each element. 

To define interior encoders, we let $E$ be a square, centred at the $(0,0) \in \mathbb{R}^2$, and formed of the union of $N_E = (2m_e+1)^2$ elements $\mathcal{T}$ for some $m_e \in \mathbb{N}$. In this way, the value of any interior latent variable at $(x,y) \in \Omega$ depends only on the solution values in the $(2m_e+1)^2$ elements symmetrically surrounding $(x,y) \in \Omega$. 

Boundary latent variables $\ell(\bs{z})$ are defined at the centre of each exterior element, using the construction described in Section \ref{sec:num_bdry_gf}, with $E_\partial$ chosen to be a line segment, centred at $0 \in \mathbb{R}$, and formed of $N_{E_\partial}=2m_e + 1$ line segments whose lengths are equal to the side length of the tile $\mathcal{T}$.  Consequently, boundary latent variables $\ell{(\bs{z})}$ depend on the boundary values $g(\bs{z})$ on the $2m_e+1$ tile boundaries  symmetrically surrounding $\bs{z} \in \partial \Omega$. Note that is has been assumed for simplicity that $N_E = N_{E_\partial}^2$. 

Decoders are created using the construction in \S\ref{sec:num_decoder_gf} by letting $D$ be a square, centred at $(0,0) \in \mathbb{R}^2$, form of the union of $N_D = (2m_d+1)^2$ tiles $\mathcal{T}$. Consequently, decoders seek to use latent values at a point $(x,y) \in \Omega$ to predict the solution   $u(x,y)$ on $N_D$ tiles symmetrically surrounding $(x,y)$. Since the decoder output shape is square the domain can be fully tiled by the decoder outputs, we use partition decoders as descibed in \S \ref{sec:decoder}.

%
As previously described the training was split into batches where each batch needs to contain samples from both the interior and domain boundary. Since the domain is square a distinction will also be made between edge training patches which do not include the corner points:
\[
(x_c,y_c) \in \left\{ (0,0),(0,1), (1,0), (1,1) \right\}
\]
and corner training patches which include a corner point. Each batch contained 128 internal samples, 128 edge samples and 32 corner samples which were equally divided between the 4 sides of the square. 

Finally, to implement the symmetric matrix $\mathcal{A}$ as an optimisation variable, we define matrices $P_{11},P_{12},P_{22} \in \mathbb{R}^{r \times r}$, let $A_{ij}:=P_{ij}+P_{ij}^\top$ and form the block matrix $\mathcal{A} = (A_{ij})_{ij=1}^2 \in \mathbb{R}^{4r^2}$. The corresponding elliptic PDE for a latent variable function $\bs{\ell} : \Omega \rightarrow \mathbb{R}^r$ is then given by
\[
D_{\mathcal{A}} \bs{\ell} = A_{11} \bs{l}_{xx} + 2 A_{12} \bs{l}_{xy} + A_{22} \bs{l}_{yy} = 0.
\]

We consider the performance of SINN models for two types of boundary data. In the first, for each boundary data function $g_i$,  we only assume that the boundary encoder can access pure boundary data  via 
\begin{equation} \label{eq:no_bdry_info}
\bs{b}_i(\bs{z}) = \left(g_i(\bs{z}),\bs{n}(\bs{z}) \right), \qquad \bs{z} \in \partial \Omega,
\end{equation}
where at the corners of the square domain, the boundary vector is defined diagonally in an outward pointing manner (e.g. $(-\frac{1}{\sqrt{2}},\frac{1}{\sqrt{2}})$ for the north-west corner). Following this, we  will also consider the case in which for each boundary data function $g_i$, the boundary derivative of its associated solution $u_i$ is available to the encoder, by letting
\begin{equation} \label{eq:bdry_info}
\bs{b}_i(\bs{z}) = \left(g_i(\bs{z}),\frac{\partial u_i}{\partial \bs{n}}(\bs{z}), \bs{n}(\bs{z}) \right), \qquad \bs{z} \in \partial \Omega,
\end{equation}
The code available for generating results is available at: https://github.com/jh6220/SINNs-for-boundary-observtion-problems.git

\subsubsection{SINN performance with pure boundary data \eqref{eq:no_bdry_info}}

Table \ref{tab: ResultsLaplaceStandardBC} shows, for different choices of the latent space dimension $r$ and the encoder and decoder complexities $N_E,N_D$, the mean square error between the $N_\text{test}$ test boundary functions solved by the SINN model $\mathcal{F}$ and ground truth is 
\[
\mathcal{E} = \frac{1}{N_\text{test}} \sum_{i=1}^{N_\text{test}} \|u_i-\mathcal{F}(\bs{b}_i) \|^2_{L^2(\Omega,\mathbb{R})}
\]
In the above equation, integrals are interpreted as sums over the $38^2$ square elements comprising the domain $\Omega$. The results in Table \ref{tab: ResultsLaplaceStandardBC} use SINNs in which each component (encoder, decoder, boundary encoder) is parameterised using a neural network with $5$ hidden layers of $60$ nodes.

\begin{table}[!ht]
\small
\caption{Mean square error $\mathcal{E}$ of SINNs using pure boundary data $\bs{b} \in (u_{|\partial \Omega},\bs{n}_{|\partial \Omega})$) for the nonlinear heat equation \eqref{eq:heat_nl}.}
\centering
\resizebox{1\columnwidth}{!}{
\begin{tabular}{c c c c cc c c c} 
\toprule
$r$   & 1     & 2     & 3   &8& 2    & 3     & 5 & 8      \\
$N_{E}$  & 1     & 3     & 3  &3& 3   & 3     & 3 & 3      \\
$N_{D}$ & 1     & 1     & 1   &1& 3  & 3     & 3  & 3    \\
 $r/N_{D}^2$& 1& 2& 3& 8& 2/9& 1/3& 5/9&8/9\\
\midrule
$\mathcal{E}$ & \num{4.54e-3}& \num{7.38e-4} & \num{5.3e-4}   &\num{6.56e-5}& \num{2.18e-2} & \num{8.14e-4} & \num{6.74e-4} & \num{5.29e-4} \\ 
\bottomrule
\end{tabular}
}
\label{tab: ResultsLaplaceStandardBC}
\end{table}

It is evident that increasing either the latent dimension $r$ or the encoder complexity $N_E$ reduces the SINN error $\mathcal{E}$. The former allows for a more complex latent space, while the latter effectively allows the encoder to access higher order derivatives of the underlying data. For example, when $N_E=3$ an encoder has access to nine local function values and is therefore has the potential to access approximate second-order derivatives.  Conversely, increasing the decoder dimension $N_D$ increases the error $\mathcal{E}$ which occurs due to the choice of partition decoder used for this example. In this case, a increasing $N_D$ corresponds to requiring the decoder to extrapolate to a larger sets, naturally increasing $\mathcal{E}$. However, if the number of degrees of freedom, i.e. $r/N_D^2$, of a trained SINN are considered it can be seen from the penultimate row of Table \ref{tab: ResultsLaplaceStandardBC} that a higher value of $N_D$ can possibly be viewed as computationally advantageous. 

To discuss the influence of latent variable dimension $r$,  we consider the trained internal elliptic models $D_\mathcal{A}$ in two cases. In the simplest case $(r=1,N_E=1,N_D=1)$, the internal elliptic model is 
\[
1.608 \, \ell_{xx}+ 0.001 \, \ell_{xy}+ 1.613 \,  \ell_{yy}=0,
\]
which is very close to the linearised PDE $\Delta u=0$. On the other hand, when more modelling degrees of freedom are available for the case $(r=3,N_E=3,N_D=1)$, the trained internal elliptic system 
\[
    \left(\begin{smallmatrix}
        1.46 & 0.49 & -1.14 \\
        0.49 & 1.57 & -0.72 \\
        -1.14 & -0.72 & 0.36 
    \end{smallmatrix}\right) \ell_{xx}   +
    \left(\begin{smallmatrix}
        0.01 & -0.03 & -0.02 \\
        -0.03 & 0.02 & -0.02 \\
        -0.02 & -0.02 & 0.06 
    \end{smallmatrix}\right) \ell_{xy}  + 
    \left(\begin{smallmatrix}
        0.92 & 1.66 & -1.30 \\
        1.66 & 0.74 & -1.58 \\
        -1.30 & -1.58 & 1.69 
    \end{smallmatrix}\right) \ell_{yy} = 0
\]
of the SINN is non-trivial, and the resulting error $\mathcal{E}$ is an order of magnitude lower than that of the simplest model.

\begin{figure}
\centering
\subfigure[$\epsilon u$]{
\includegraphics[width=0.3\textwidth,clip=true,trim=0cm 0cm 0cm 0cm]{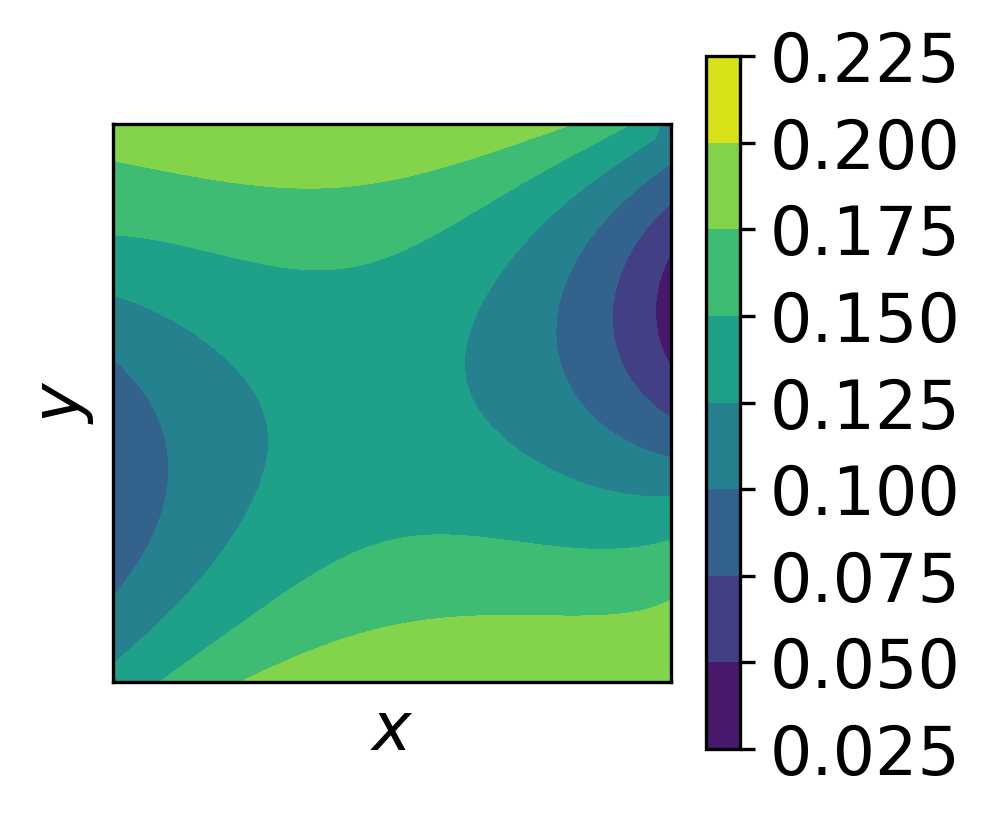}   
}
\subfigure[$\left(\mathcal{E}_\mathcal{A} \circ \epsilon^{\partial} \right)\bs{b}$]{
\includegraphics[width=0.3\textwidth,clip=true,trim=0cm 0cm 0cm 0cm]{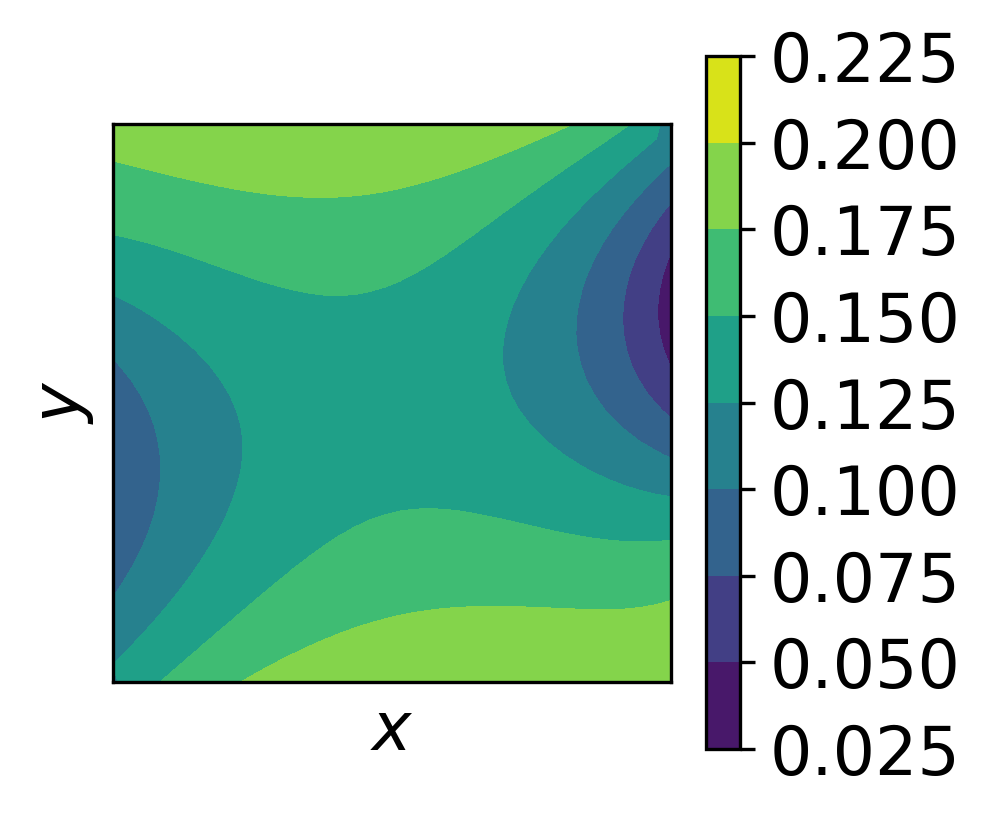}   
}
\subfigure[$\epsilon u-\left(\mathcal{E}_\mathcal{A} \circ \epsilon^{\partial} \right)\bs{b}$]{
\includegraphics[width=0.32\textwidth,clip=true,trim=0cm 0cm 0cm 0cm]{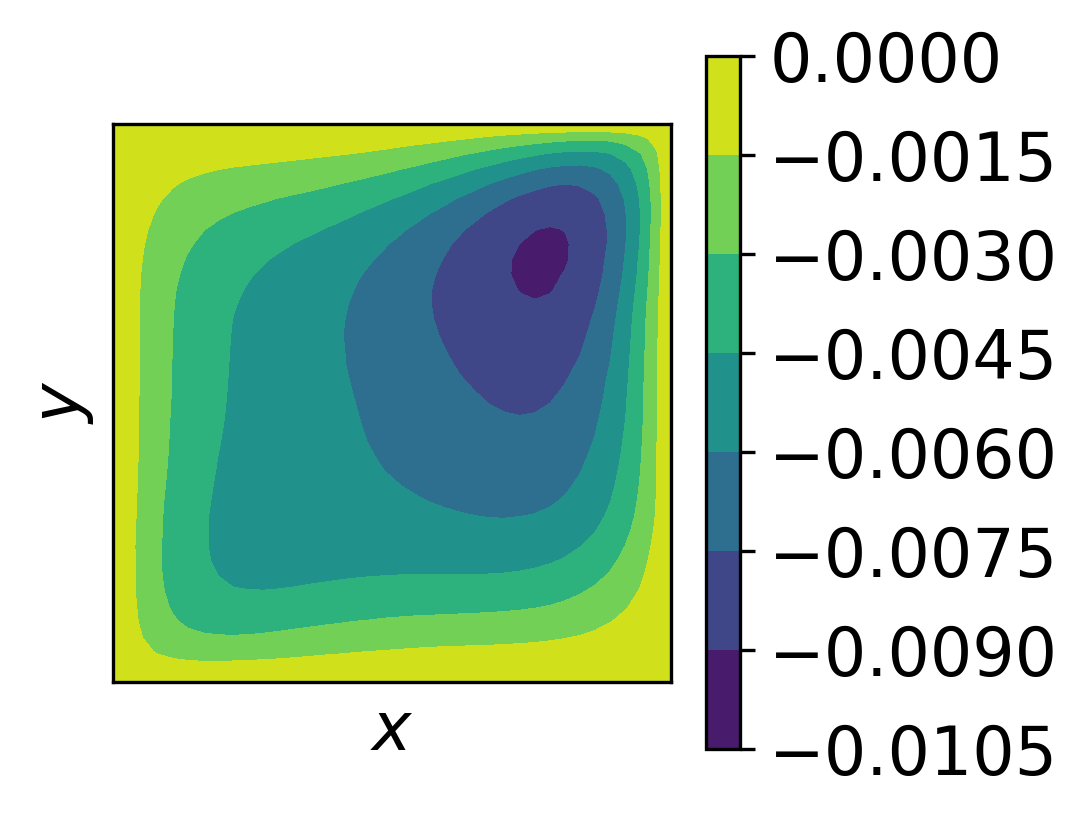}   
}
\\
\subfigure[$u$]{
\includegraphics[width=0.3\textwidth,clip=true,trim=0cm 0cm 0cm 0cm]{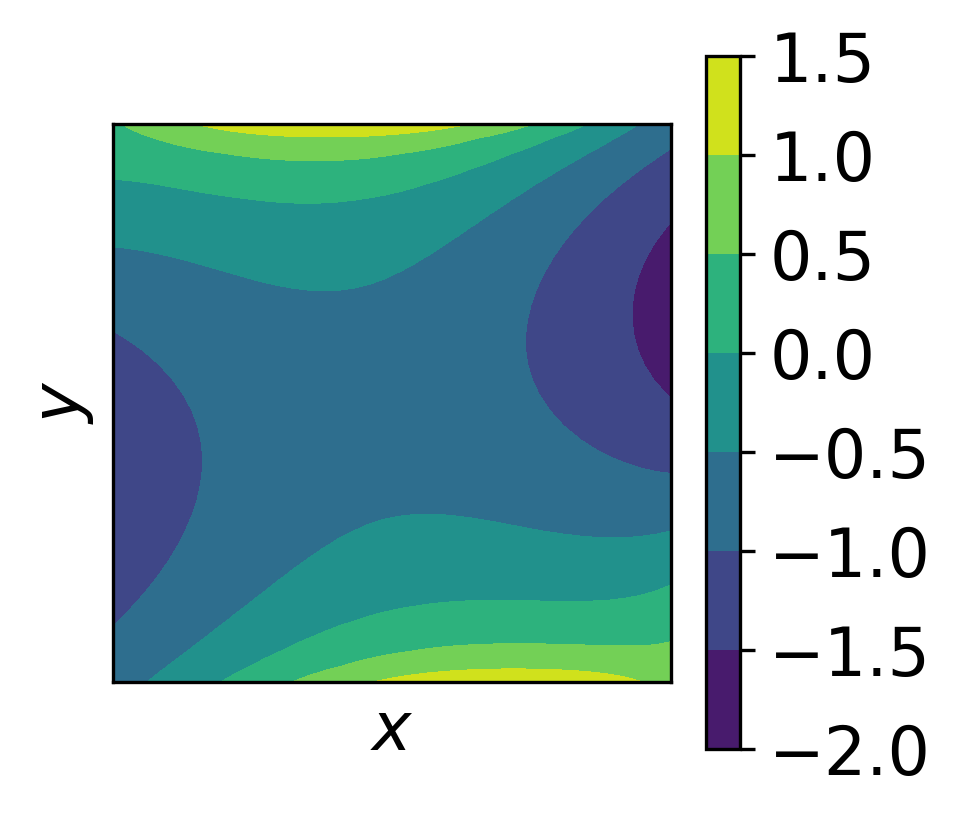}   
}
\subfigure[$\mathcal{F}(\bs{b})$]{
\includegraphics[width=0.3\textwidth,clip=true,trim=0cm 0cm 0cm 0cm]{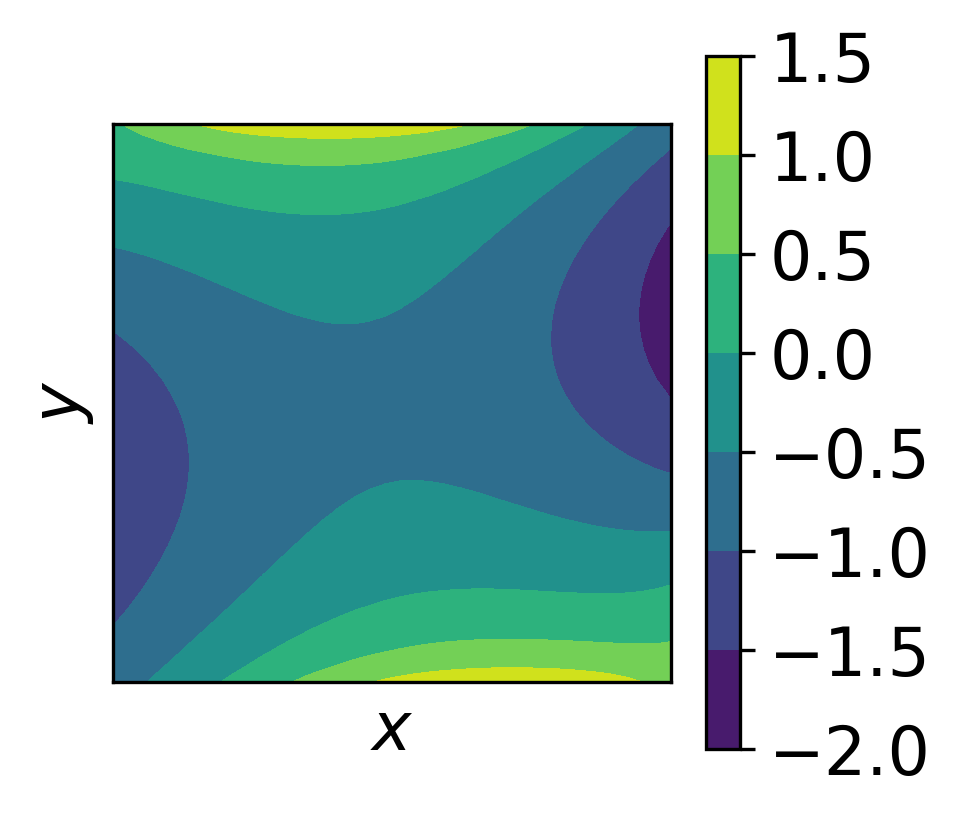}   
}
\subfigure[$\mathcal{F}(\bs{b})-u$]{
\includegraphics[width=0.32\textwidth,clip=true,trim=0cm 0cm 0cm 0cm]{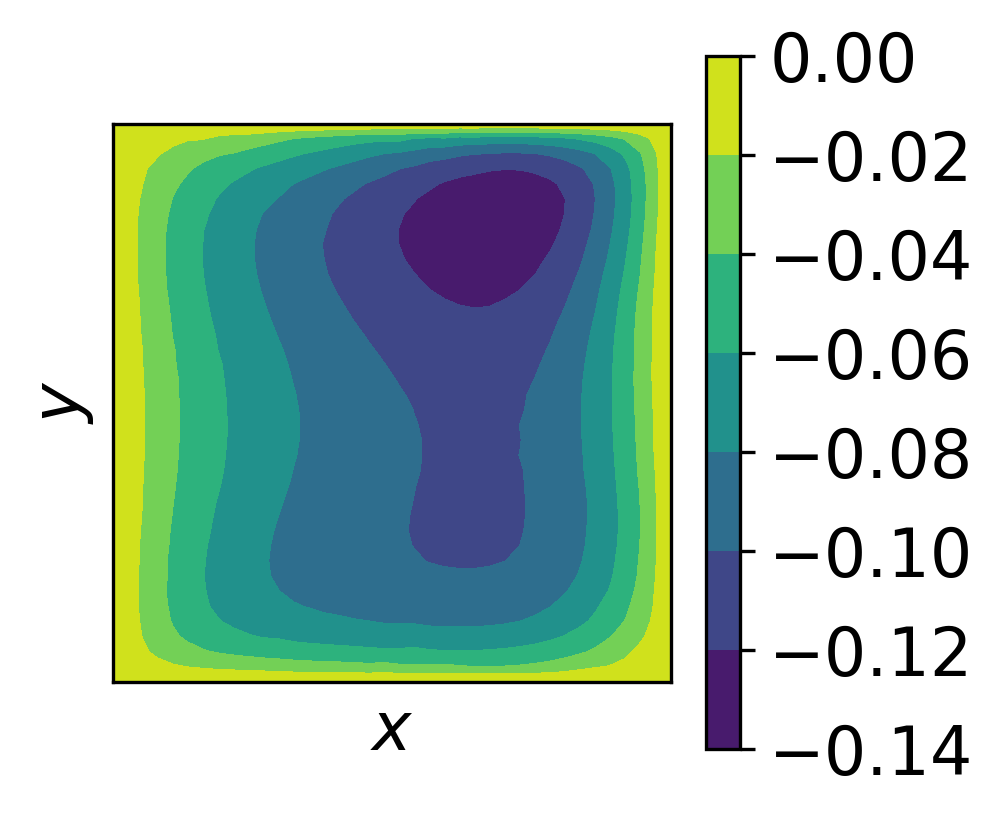}   
}
\caption{Indicative example of SINN performance with $(r=1,N_E=1,N_D=1)$. (a) shows the encoded latent variable $\epsilon u$; (b) the solved latent variable $\left(\mathcal{E}_\mathcal{A} \circ \epsilon^{\partial} \right)\bs{b}$; (c) latent variable error; (d) original data $u$; (e) SINN solution; (f) SINN error. The SINN error for this data pair is comparable to the mean value reported in Table \ref{tab: ResultsLaplaceStandardBC}}
\label{fig:Model111_contourplot}
\end{figure}


We next discuss the structure of the trained encoder and boundary encoders. For the simplest case $(r=1,N_E=1,N_D=1)$, Figure \ref{fig:Model111_contourplot} visualises, for a selected test data pair $(u,\bs{b})$, the SINN solution $\mathcal{F}(\bs{b})$ computed using the boundary data in (e); the encoded latent variable $\epsilon u$ in (a); and the reconstructed latent variable $(\mathcal{E}_\mathcal{A} \circ \epsilon^\partial) \bs{b}$ in (b). The error field of both the SINN and the internal latent variables are also shown in Figure \ref{fig:Model111_contourplot}~(c),(f) and these show small, but non-trivial, discrepancies.

Since the trained PDE is approximately equivalent the linearised PDE, and the latent variable space has scalar-values $(r=1)$, it is not surprising that the encoded latent variable $\epsilon u$ and the original data $u$ are superficially similar. Model accuracy in this case is achieved purely from the nonlinearity of the boundary encoder and decoder. Indeed, Figure \ref{fig:Model111_sliceplot}~(a--b) shows two slices of the SINN solution $\mathcal{F}(\bs{b})$, ground-truth data $u$, and the solution obtained by extending the boundary data using just the linearised PDE $\Delta u=0$. It is clear that the nonlinear SINN solution is substantially more accurate than the linearised model. To understand the precise way in which the nonlinear structure achieves this increase in accuracy, Figure  \ref{fig:Model111_sliceplot}~(c) shows both the true boundary data $\bs{b}$ and the encoded data $\epsilon^\partial \bs{b}$. The encoder $\epsilon^\partial$ behaves assymetricaly in the sense that it attenuates positive boundary values and amplifies negative ones. The reason for this behaviour is that the local diffusion coefficient $e^{u(x,y)}$ of the PDE \eqref{eq:heat_nl} increases exponentially with the $u(x,y)$. Thus, positive boundary values imply diffusion on shorter length scales compared to negative boundary values. The boundary encoder's behaviour can now be interpreted as $\epsilon^\partial$ reflecting this local nonlinear structure of the underlying nonlinear diffusion coefficient, and this improves the accuracy of the SINN operator $\mathcal{F}$. While this behaviour of the  boundary encoder is now interpretable, there is a persistent error if a simple model with $r=1$ is used. As indicated in Table \ref{tab: ResultsLaplaceStandardBC}, this error can be avoided by increasing the dimension $r$ of the latent space.  

\begin{figure}
\centering
\subfigure[]{
\includegraphics[width=0.30\textwidth,clip=true,trim=0cm 0cm 0cm 0cm]{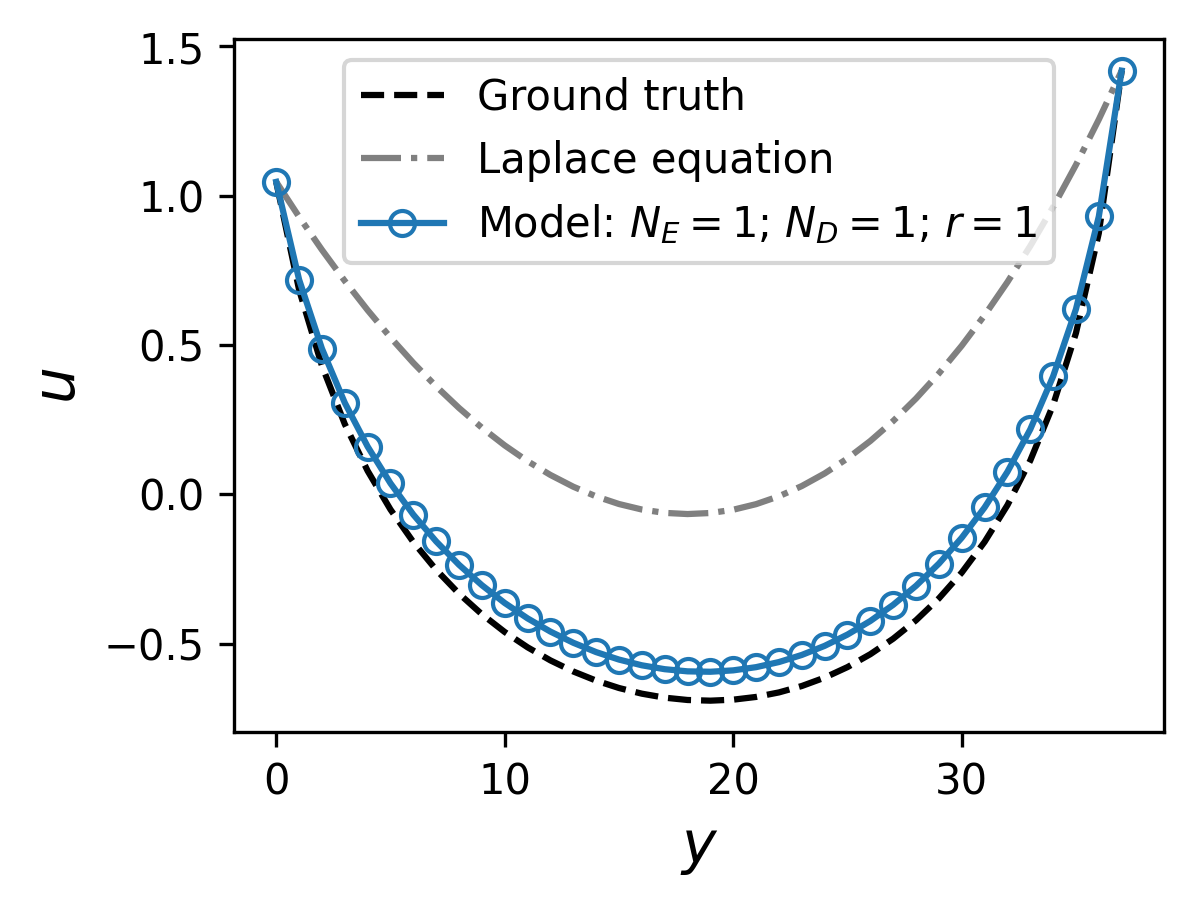}   
}
\subfigure[]{
\includegraphics[width=0.30\textwidth,clip=true,trim=0cm 0cm 0cm 0cm]{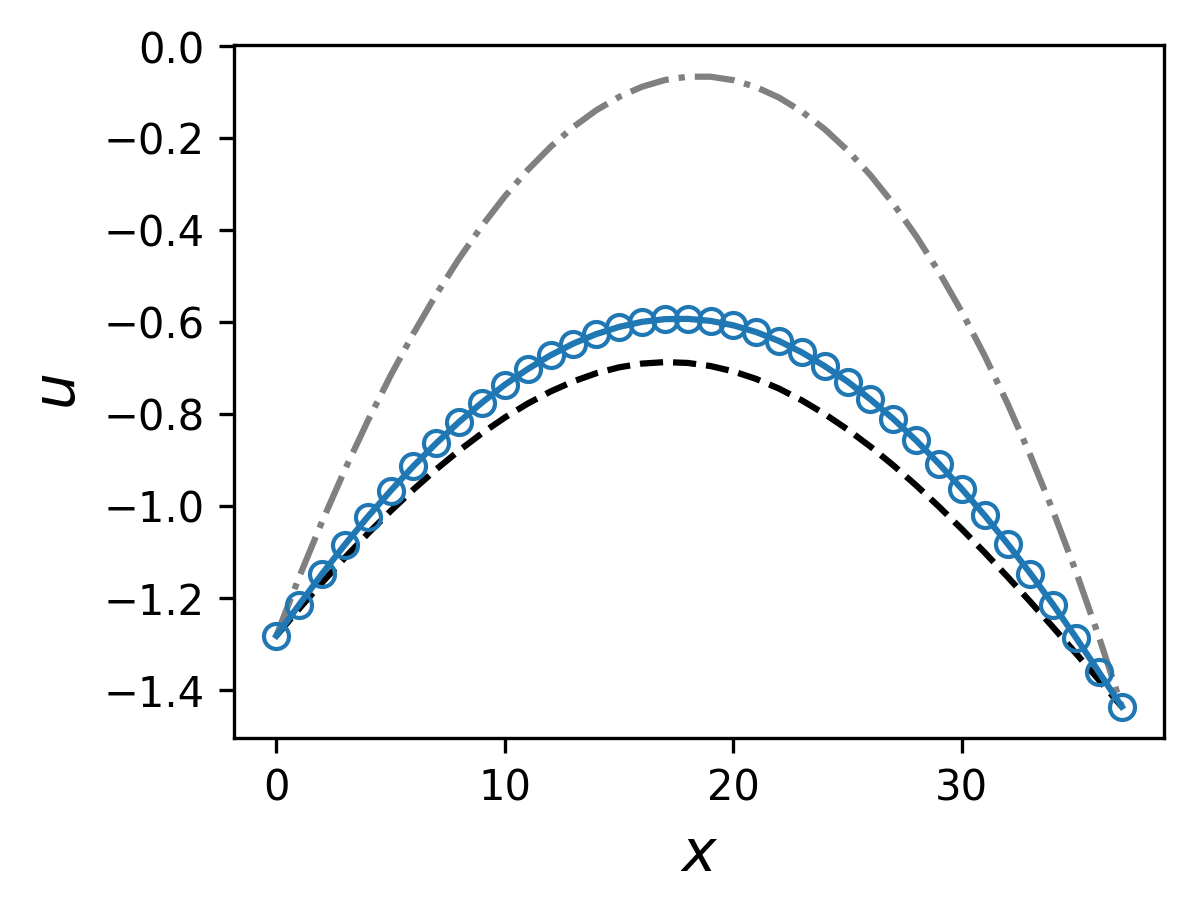}   
}
\subfigure[]{
\includegraphics[width=0.30\textwidth,clip=true,trim=0cm 0cm 0cm 0cm]{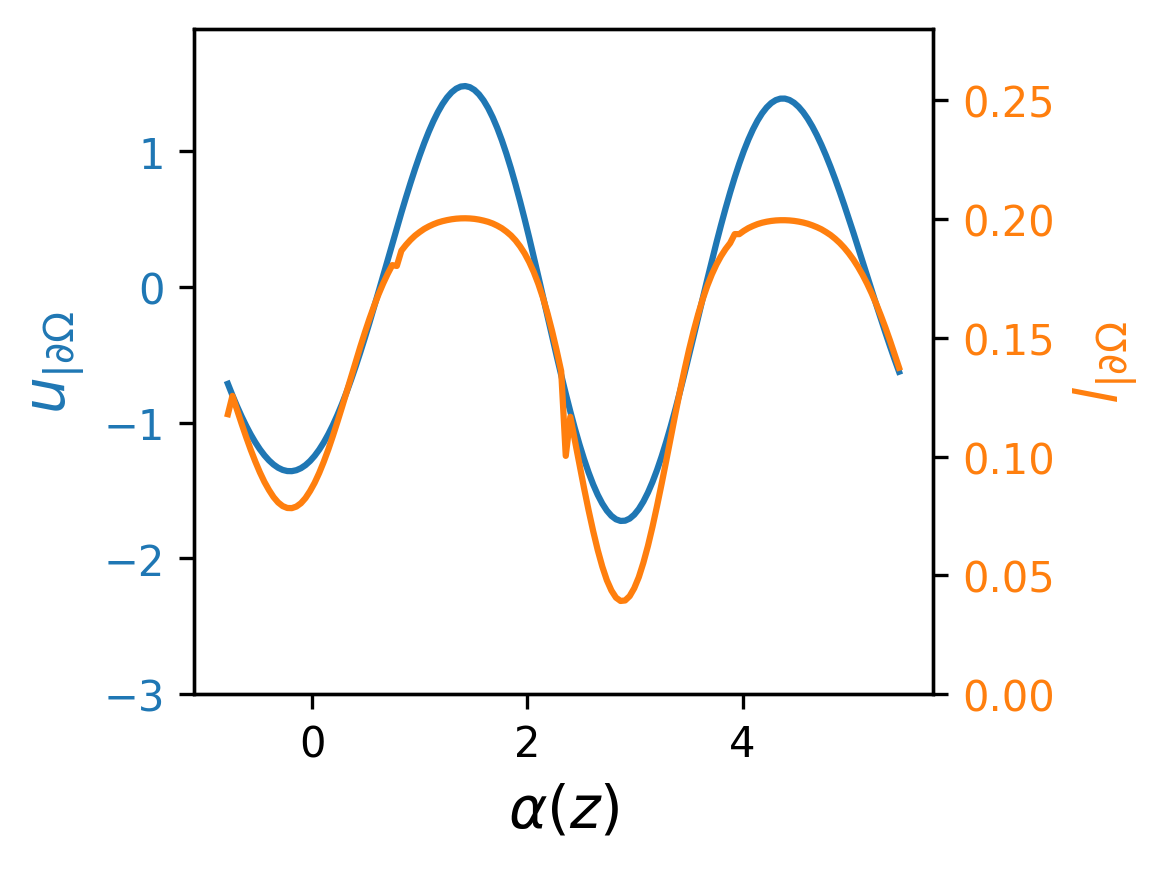}   
}
\caption{Indicative example of SINN performance for $(r=1,N_E=1,N_D=1)$. (a) shows solutions on a slice through the domain at $x=0.5$; (b) the respective variables on the slice $y=0.5$; (c) Original and encoded boundary data.}
\label{fig:Model111_sliceplot}
\end{figure}


Finally, we seek to understand the geometric properties of the latent variables for cases in which $r=3$ and $r=5$. While the latent variables do not have an strict physical meaning, one can use sensitive analysis to extract their underlying structure. In particular, Figure \ref{fig:StandardBC_sensitivity} shows 
\[
\frac{\partial \delta}{\partial \bs{\ell}} (\bar{\bs{\ell}})
\]
where $\bar{\ell}$ is the mean value of the latent variables computed across the entire data ensemble. The gradient was computed numerically using central finite difference method with a step size equal to a standard deviation of each latent dimension computed from the data-set in a similar manner to the mean.

\begin{figure}
\centering
\includegraphics[width=1\textwidth,clip=true,trim=0cm 0cm 0cm 0cm]{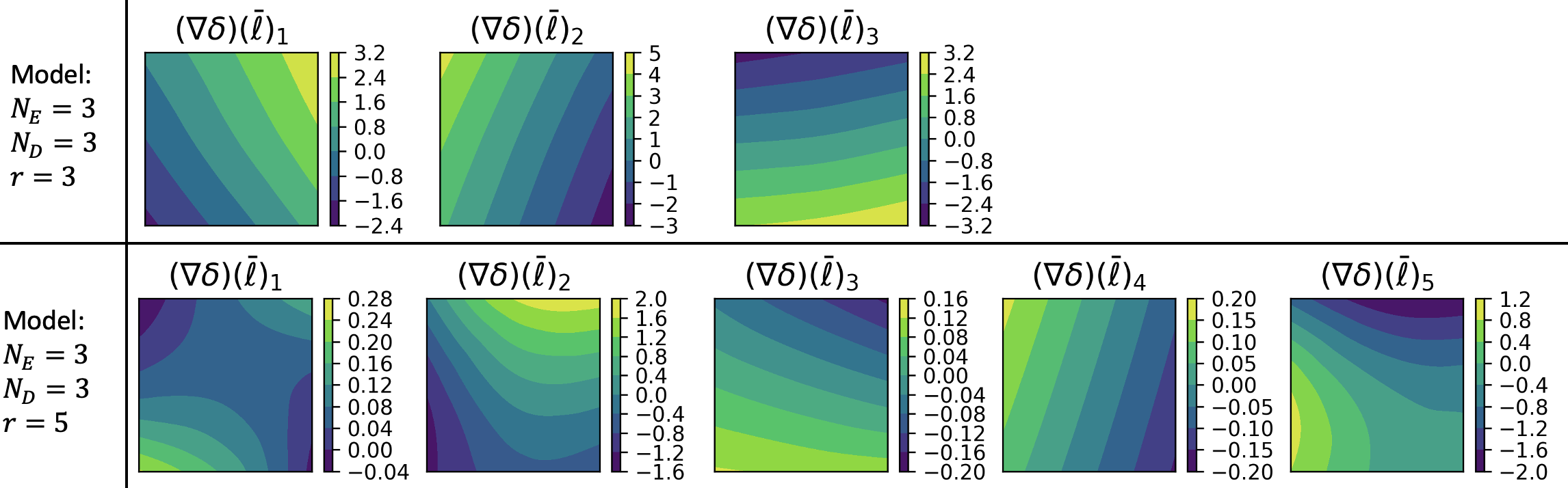}    
\caption{Latent dimensions interpretation using sensitivity analysis for models with pure boundary data $\bs{b} \in (u_{|\partial \Omega},\bs{n}_{|\partial \Omega})$). The figure shows the gradients $(\nabla \delta)(\bar{\bs{\ell}})_{i} = \frac{\partial \delta}{\partial \bs{\ell}_{i}}(\overline{\bs{\ell}})$ interpolated onto a fine grid}
\label{fig:StandardBC_sensitivity}
\end{figure}

The results of the sensitivity analysis for two with $(r=3,N_E=3,N_D=3)$ and $(r=5,N_E=3,N_D=3)$ are shown in Figure \ref{fig:StandardBC_sensitivity}. It can be observed that the implies latent variables are spatially coherent. In the case $r=3$, the structures are approximately orthogonal linear surfaces, while for $r=5$ more complex, yet still coherent, spatial structures can be observed. 

\subsubsection{SINN performance with extended boundary data \eqref{eq:bdry_info}}

We now consider the case in which extra boundary data, namely the normal boundary derivatives, are available to the boundary encoder. Table \ref{tab: ResultsLaplaceStandardBC} shows SINN errors $\mathcal{E}$ for a variety of model choices. These follow a similar trend to the case of pure boundary conditions, although the final row of Table \ref{tab: ResultsFluidExtraBC} indicates that extra boundary information gives a consistent performance improvement, and that this improvement is more pronounced for higher values of $N_D$ and of $r$.

\begin{table}[!ht]
\small
\caption{SINN errors $\mathcal{E}$ when extra boundary information ($\bs{b} = (u_{|\partial \Omega},\bs{n}_{|\partial \Omega}, \frac{\partial u}{\partial \bs{n}}_{|\partial \Omega})$) is used.}
\centering
\resizebox{\columnwidth}{!}{
\begin{tabular}{c c c c c c c}
\toprule
$r$                & 2     & 3 & 8  & 2    & 3     & 5      \\
$N_{E}$      & 3     & 3 & 3 & 3   & 3     & 3      \\
$N_{D}$      & 1     & 1 & 1  & 3  & 3     & 3      \\ 
\midrule
 $\mathcal{E}$ & \num{6.63e-4} & \num{2.50e-4} & \num{1.88e-5}  & \num{7.3e-3} & \num{3.41e-4} & \num{1.43e-4} \\ 
\hline
  $\mathcal{E} / \mathcal{E}_{\text{pure BCs}}$ & 1.11& 2.11& 3.49& 1.78& 2.37&4.68\\
\end{tabular}
}
\label{tab: ResultsLaplaceExtraBC}
\end{table}

Figure \ref{fig:ExtraBC_sliceplots} shows, for an  test function $u$ whose error is indicative of the mean values presented in Table \ref{tab: ResultsFluidExtraBC}, slices thought the true solution and SINN solutions. Both slices at $x=0$ and $y=0$ in Figure \ref{fig:ExtraBC_sliceplots}~(a--b) show very good agreement of the SINN solution $\mathcal{F}(\bs{b})$ with the true solution $u$. Consider first the SINN models with $N_D=1$, whose decoders are not required to extrapolate. It can be seen in both error plots of Figure \ref{fig:ExtraBC_sliceplots}~(c--d) that the pointwise error decreases uniformly as $r$ increases. Next, consider the SINN models with $N_D=3$, whose decoders must extrapolate to two adjacent elements. In this case, while absolute error also decreases with increasing $r$, the error plots are oscillatory as a result of the extrapolation error involved with the chosen partition decoder.

%


\begin{figure}
\centering
\subfigure[]{
\includegraphics[width=0.45\textwidth,clip=true,trim=0cm 0cm 0cm 0cm]{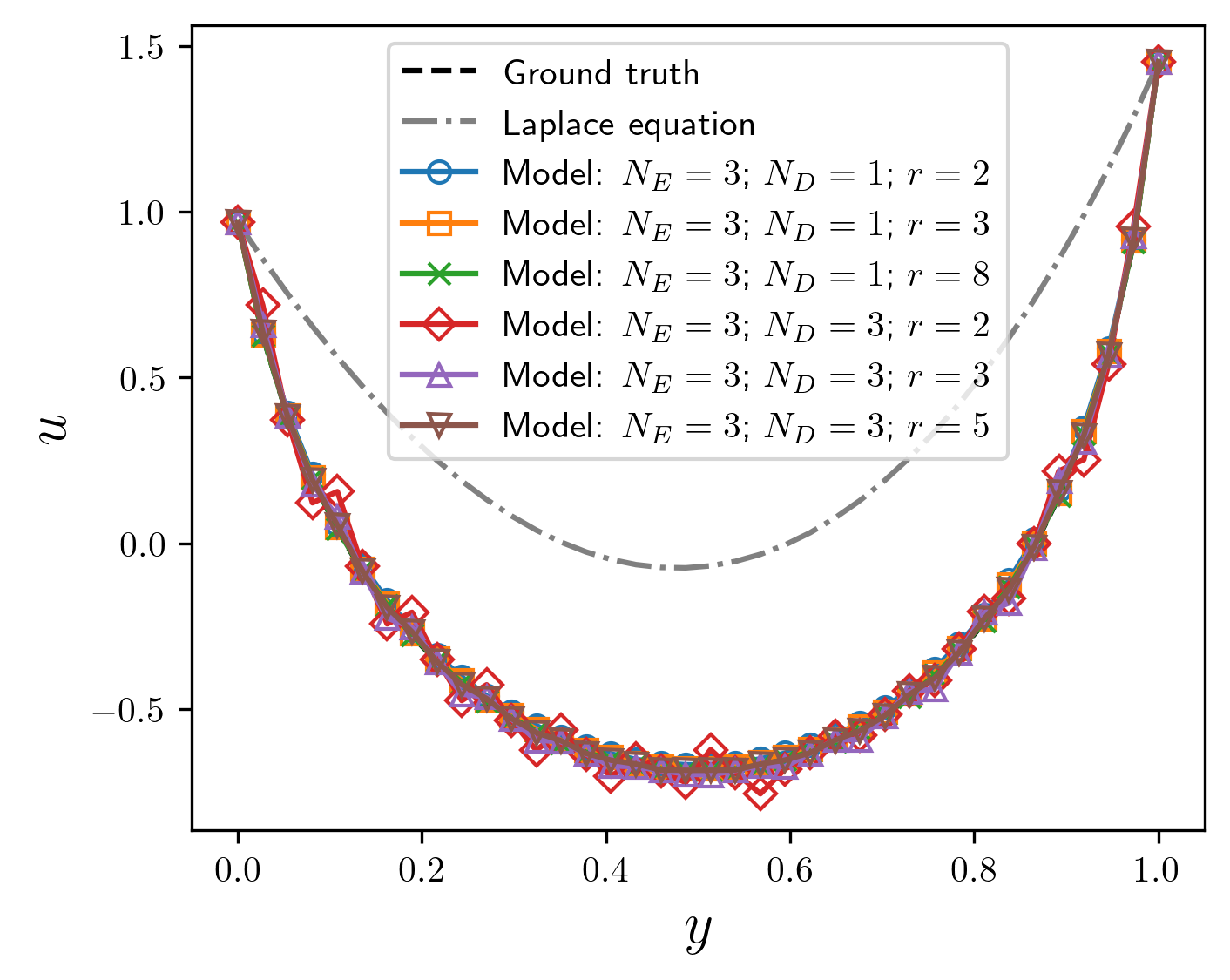}   
}
\subfigure[]{
\includegraphics[width=0.45\textwidth,clip=true,trim=0cm 0cm 0cm 0cm]{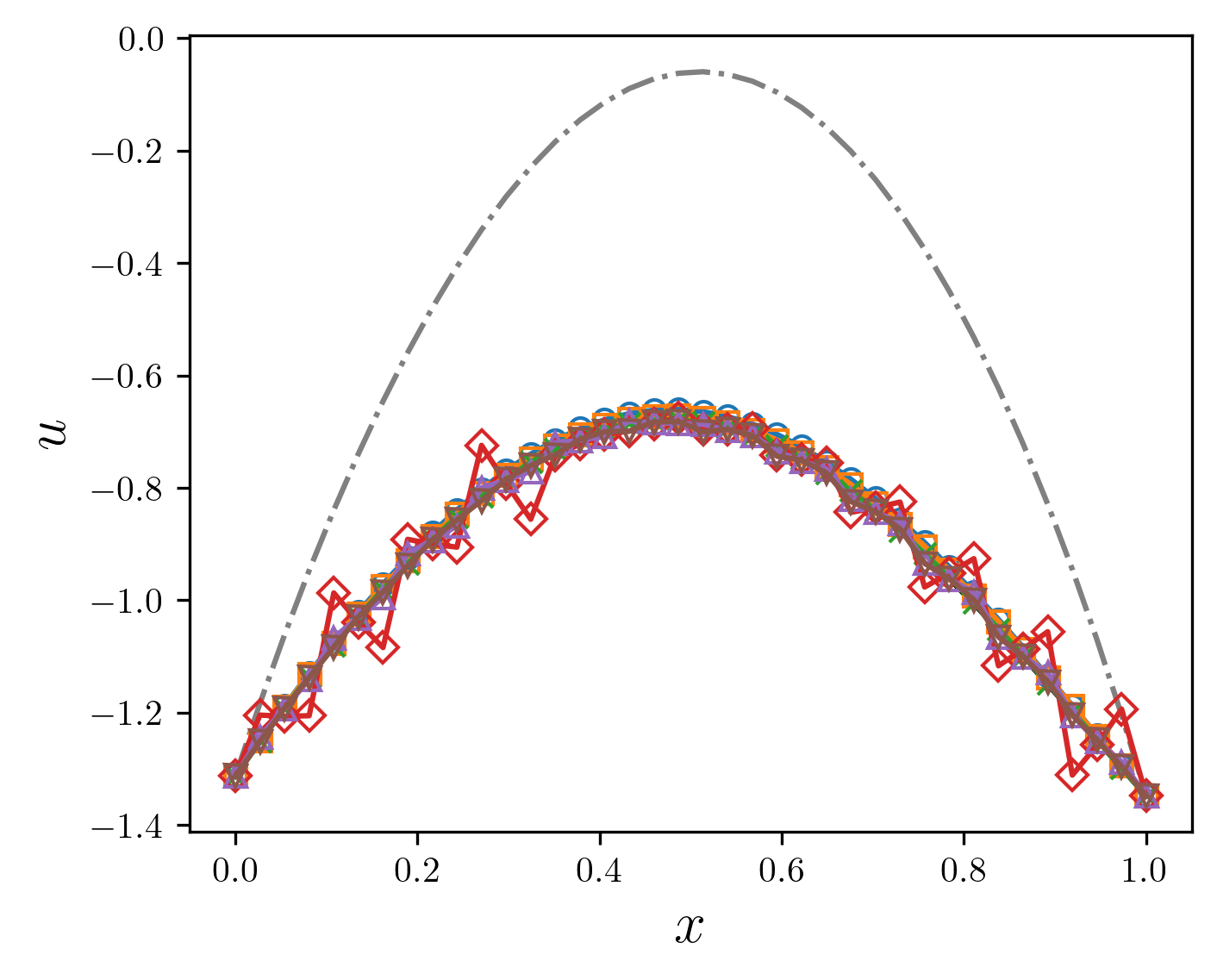}   
}
\\
\subfigure[]{
\includegraphics[width=0.45\textwidth,clip=true,trim=0cm 0cm 0cm 0cm]{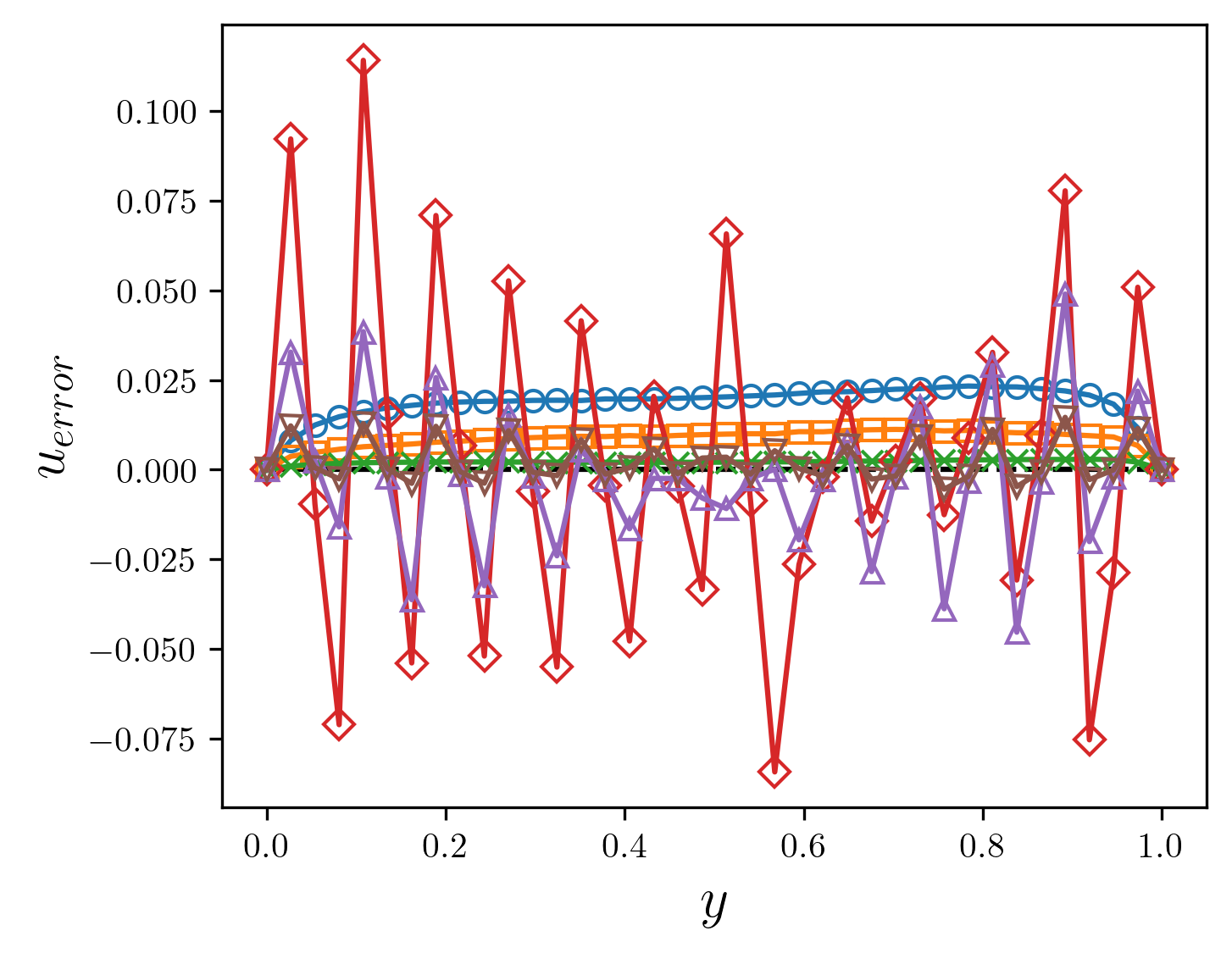}   
}
\subfigure[]{
\includegraphics[width=0.45\textwidth,clip=true,trim=0cm 0cm 0cm 0cm]{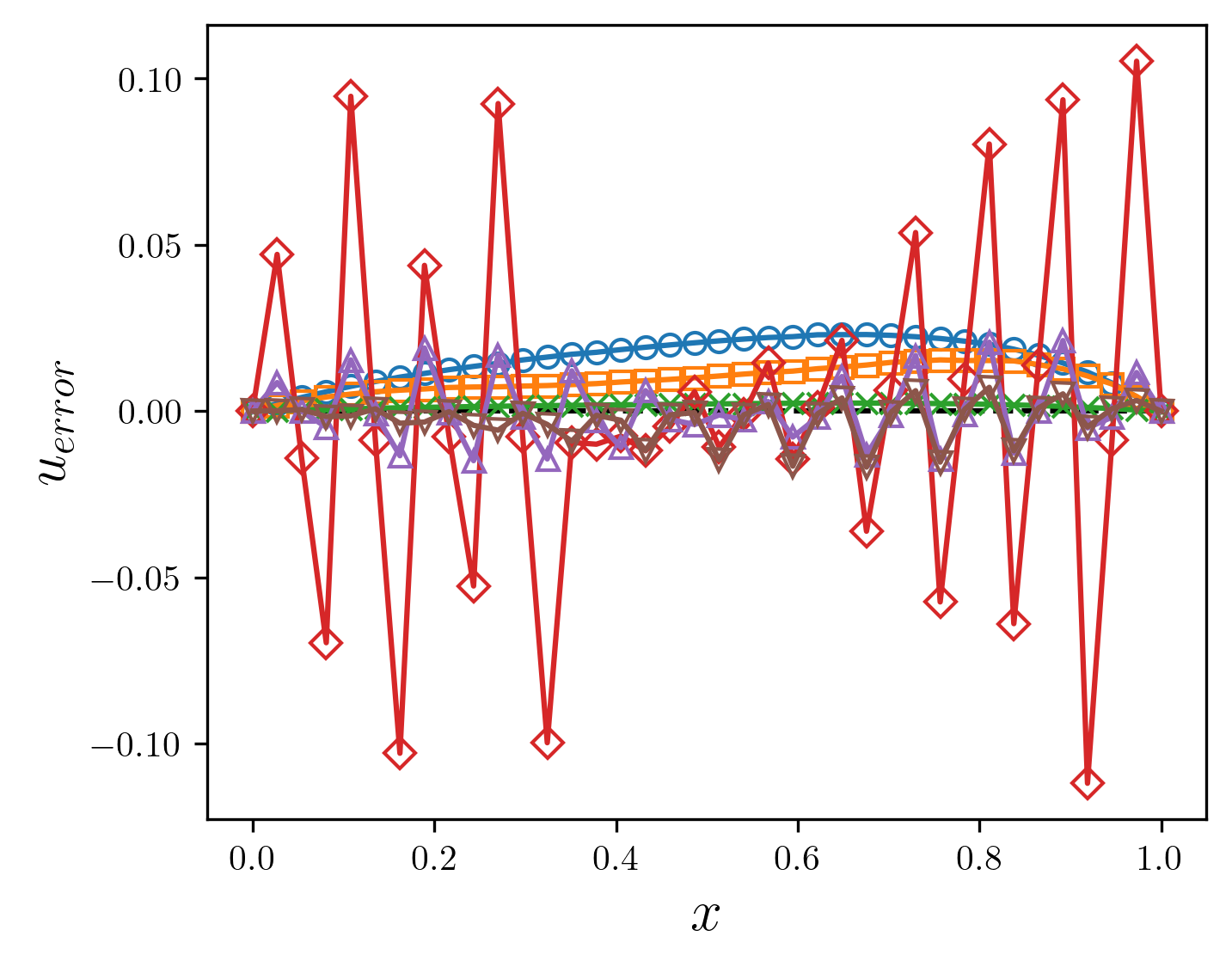}
}
\caption{An indicative example of SINN performance for the parameters shown in Table \ref{tab: ResultsLaplaceExtraBC}. Results are plotted along slices through the domain where $x=0.5$ (left) and $y=0.5$ (right). The upper plots (a-b) show the test data $u$ and SINN reconstructions $\mathcal{F}(\bs{b})$; the lower plots (c--d) show errors $u_\text{error} = u - \mathcal{F}(\bs{b})$.}
\label{fig:ExtraBC_sliceplots}
\end{figure}


Finally, we discuss the influence of underlying neural network complexity on the SINN error. For the case $(r=5,N_E=3,N_D=3)$, Table \ref{tab: ResultsLaplaceParameterSize} shows the average SINN error $\mathcal{E}$ over the testing ensemble for different choices of Neural Network dimensions. Three different neural network structures are considered, in terms of the number of layers and nodes per layer, with each case applying to the encoder, boundary encoder and decoder.  It is interesting to note that $\mathcal{E}$ decreases with neural network complexity, suggesting the model over fitting has not occurred for these parametric values. This highlights a potential benefit of the SINN methodology in that, due to the use of training patches, significant training information can be obtained from each training data pair $(\bs{u}_i,\bs{b}_i)$. Consequently, the SINN approach appears robust to overfitting, even when employing only a relatively small training data ensemble.

\begin{table}[!ht]
\small
\caption{SINN error $\mathcal{E}$ as a function of the neural network parameters for the case $(r=5,N_E=3,N_D=3)$.}
\centering
\resizebox{0.7\columnwidth}{!}{
\begin{tabular}{c c c c} 
\toprule

$N_\text{layers}$            & 4     & 5     & 5      \\
$N_\text{nodes}$     & 40     & 60     & 200      \\
$N_\text{parameters}$ & \num{5.53e3}     & \num{1.55e4}     & \num{1.64e5}      \\
\midrule
$\mathcal{E}$& \num{1.03e-3} & \num{1.43e-4} & \num{2.34e-5} \\ 
\hline
\end{tabular}
}
\label{tab: ResultsLaplaceParameterSize}
\end{table}

\subsection{Steady laminar fluid flow}

Let $\Omega = [0,1] \times [0,1] \subset \mathbb{R}^2$ and suppose that $\bs{u} : \Omega \rightarrow \mathbb{R}^2$ and $p:\Omega \rightarrow \mathbb{R}$ satisfy the steady, incompressible, Navier-Stokes equations 
\begin{align}
\begin{split}
\bs{u} \cdot \nabla \bs{u} + \nabla p&=  \nu \Delta \bs{u}, \qquad \text{in} \; \Omega\\
\nabla \cdot \bs{u} &= 0, \phantom{\Delta u} \qquad\; \text{in} \; \Omega,\\
\bs{u} &= \bs{g},\phantom{\Delta u} \qquad \; \text{on} \; \partial \Omega. 
\end{split} \label{eq:LaminarFLow}
\end{align}
Here, $\bs{u} = (u_x,u_y)$ represents the velocity compents of a fluid contained in the square domain $\Omega$, $p$ is the pressure of the fluid, and $\nu >0$ is the kinematic viscosity of the fluid. In this example, $\Omega$ should be thought of a control volume in a larger fluid flow. The corresponding velocity boundary conditions must, by the divergence theorem and incompressibility, then satisfy 
\begin{equation} \label{eq:boundary_data_restriction}
\int_{\partial \Omega} \bs{g} \cdot \bs{n} \, dS = \int_{\Omega} \nabla \cdot \bs{u} \, dV =0.  
\end{equation}

The aim of this example is to study whether SINNs can reconstruct the fluid velocity in the domain $\Omega$, using only boundary velocity data. This presents a more complex and challenging example than the nonlinear heat equation in the previous section, in view of the three-dimensional state space comprising of two velocity components and the pressure. Furthermore, to increase the challenge presented by this example, we will not seek to exploit any pressure information, meaning that its influence must be automatically discovered during SINN training. We note that, although we do not consider such an example here, the SINN methodology could analogously be applied to a fluid flow example with no-slip (i.e. Dirichlet) boundary conditions in which the interior fluid velocity must be recovered using only boundary pressure data.  

A data ensemble is created using a similar method to that in \S\ref{subsec:nl_heat}. Here,  $2\times 10^3$ random, sinusoidal, boundary velocity functions $\bs{g}_i = ((u_x)_i,(u_y)_i)$ are generated, each of which also satisfies the constraint \eqref{eq:boundary_data_restriction}. For each boundary data function, the PDE \eqref{eq:LaminarFLow} was solved using a SIMPLE algorithm on a staggered grid, with the domain $\Omega$ discredited into $30 \times 30$ rectangular cells, with the pressure $p$ computed at each grid centre, and the velocity components $u_x,u_y$ computed at the mid-point of the cell's sides. The chosen grid was non-uniform with greater resolution near the boundaries to aid numerical convergence. The resulting ``ground-truth'' solutions were then re-sampled onto a $38 \times 38$ uniform grid with the same properties as in \S\ref{sec:num_bdry_gf}. 


The generating functions and training loop was implemented in an equivalent manner to the example in \S\ref{subsec:nl_heat}. The only difference for this example is that encoders and boundary encoders take the two velocity components  $u_x, u_y$ as inputs. We note again, that pressure information is not available during training. Here, for brevity, we focus on case in which extra boundary information $\bs{b} =  (\bs{u}_{|\partial \Omega}=\bs{g}, \frac{\partial \bs{u}}{\partial \bs{n}}_{|\partial \Omega},\bs{n})$ is available.  Each trained SINN model this section has the same architecture for generating functions $\epsilon$, $\epsilon^{\partial}$ and uses a partition decoder $\delta$. The models with $N_D=1$ have $5$ hidden layers of $60$ nodes, and models $N_D=3$ have $5$ hidden layers of $200$ nodes to compensate for the higher dimensional decoder output.

Table \ref{tab: ResultsFluidExtraBC} shows the testing errors $\mathcal{E}$ of SINN models created with a selected parameter values. The errors exhibit the same trends as observed for the nonlinear heat equation in \S\ref{subsec:nl_heat}, with modelling error decreasing with increasing latent space dimension $r$ and encoder dimension $N_E$, and with errors increasing as the extrapolation dimension $N_D$ of the partition decoder is increased. The final row of Table \ref{tab: ResultsFluidExtraBC} shows the increase in error if only standard boundary conditions $(\bs{u}_{|_{\partial \Omega}}, \bs{n})$ are available in model training although, for brevity, we do not discuss these results in detail here.

\begin{table}[!ht]
\small
\caption{Mean square testing error $\mathcal{E}$ of SINNs using boundary information  $\bs{b} =  (\bs{u}_{|\partial \Omega}, \frac{\partial \bs{u}}{\partial \bs{n}}_{|\partial \Omega},\bs{n})$ to solve the PDE \eqref{eq:LaminarFLow}.}
\centering
\resizebox{\columnwidth}{!}{%
\begin{tabular}{c c c c c c}
\hline
$r$                & 4     & 6     & 6     & 8 & 10      \\
$N_{E}$      & 3     & 3   & 3     & 3 & 3      \\
$N_{D}$      & 1     & 1  & 3     & 3  & 3    \\ \hline
$\mathcal{E}$    & \num{2.32e-3} & \num{1.69e-3}  & \num{2.61e-3} & \num{8.32e-4} & \num{6.75e-4} \\ 
\hline
 $\mathcal{E} / \mathcal{E}_\text{pure BCs}$& 2& 2.15& 1.79& 3.61 & 2.65\\
\end{tabular}
}
\label{tab: ResultsFluidExtraBC}
\end{table}

An indicative visualisation of the internal structure of two of the trained SINNs is given in Figure \ref{fig:ExtraBC_contourplots_fluid} for the cases $(r=6,N_E=3,N_D=1)$ and $(r=10,N_E=3,N_D=3)$. The true test function velocity components $u_x,u_y$ are shown in the top row. Both models give a solution to the boundary value problem $\mathcal{F}(\bs{b})$ which is a very good approximation to 
 the true flow, as is shown in the right-hand column of Figure \ref{fig:ExtraBC_contourplots_fluid}. The main noticeable difference is that the SINN model with more latent variables $r=10$ is able to more accurately capture the elongated vertical structure corresponding to values $u_y <0$. Both models have coherent, yet non-trivial, latent variables which are shown in the middle two columns of Figure \ref{fig:ExtraBC_contourplots_fluid}. The additional degrees of freedom enjoyed by the SINN with $r=10$ allows for more accurate reconstruction of the finer scale flow features than the simpler model with $r=6$.

\begin{figure}
\centering
\includegraphics[width=1\textwidth,clip=true,trim=0cm 0cm 0cm 0cm]{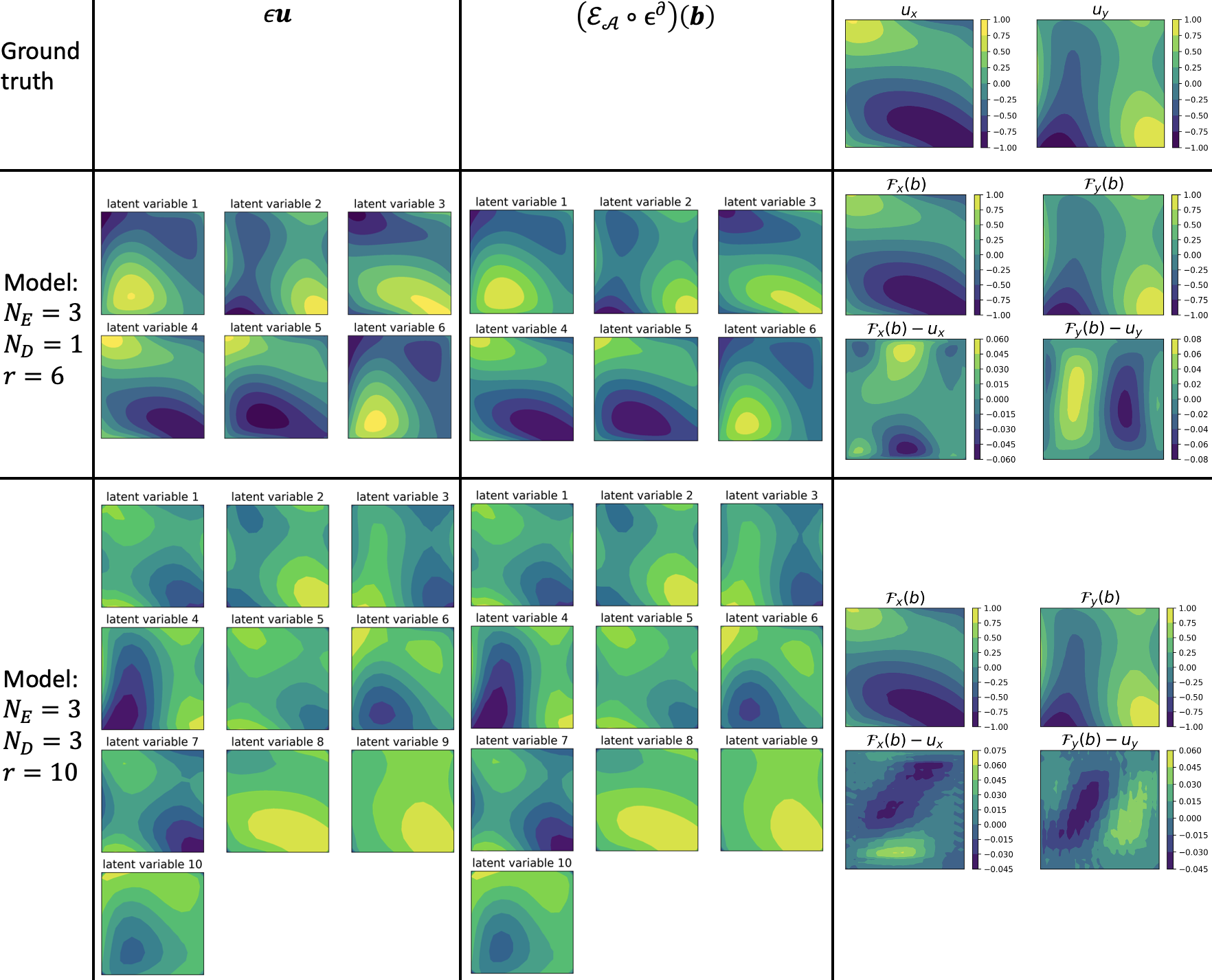}
\caption{An indicative example of SINN performance and latent variables when solving the boundary value problem \eqref{eq:LaminarFLow}. The true solution, $\bs{u}$, is shown in the top-right corner, and its SINN approximations and error fields for the two indicated models shown in the right-hand column. The encoded latent variables $\epsilon \bs{u}$ and solved latent variables $(\mathcal{E}_\mathcal{A} \circ \epsilon^\partial ) (\bs{b})$ are shown in the left-hand and middle columns, respectively. }
\label{fig:ExtraBC_contourplots_fluid}
\end{figure}


To look more closely at SINN performance for boundary observation of the Navier-Stokes PDE \eqref{eq:LaminarFLow}, Figure \ref{fig:ExtraBC_sliceplots_fluidU} shows, for one indicative example, slices through the SINN solutions $\mathcal{F}(\bs{b})$ at $x=0.5$ and $y=0.5$ for all models considered in Table \ref{tab: ResultsFluidExtraBC}. All models exhibit a good approximation to the main trends of the true data, with approximation error decreasing with increasing latent space dimension $r$. For this more challenging example, modelling errors have not yet converged for the parameter values shown. However, as can be observed from the reconstructed flow fields shown in Figures \eqref{fig:ExtraBC_contourplots_fluid}, all SINN models are able to reconstruct a very close approximation to the dominant internal vortex structures in the flow domain. Such performance, if replicated in experimental applications, for example, would be of great practical use.

\begin{figure}
\centering
\includegraphics[width=1\textwidth,clip=true,trim=0cm 0cm 0cm 0cm]{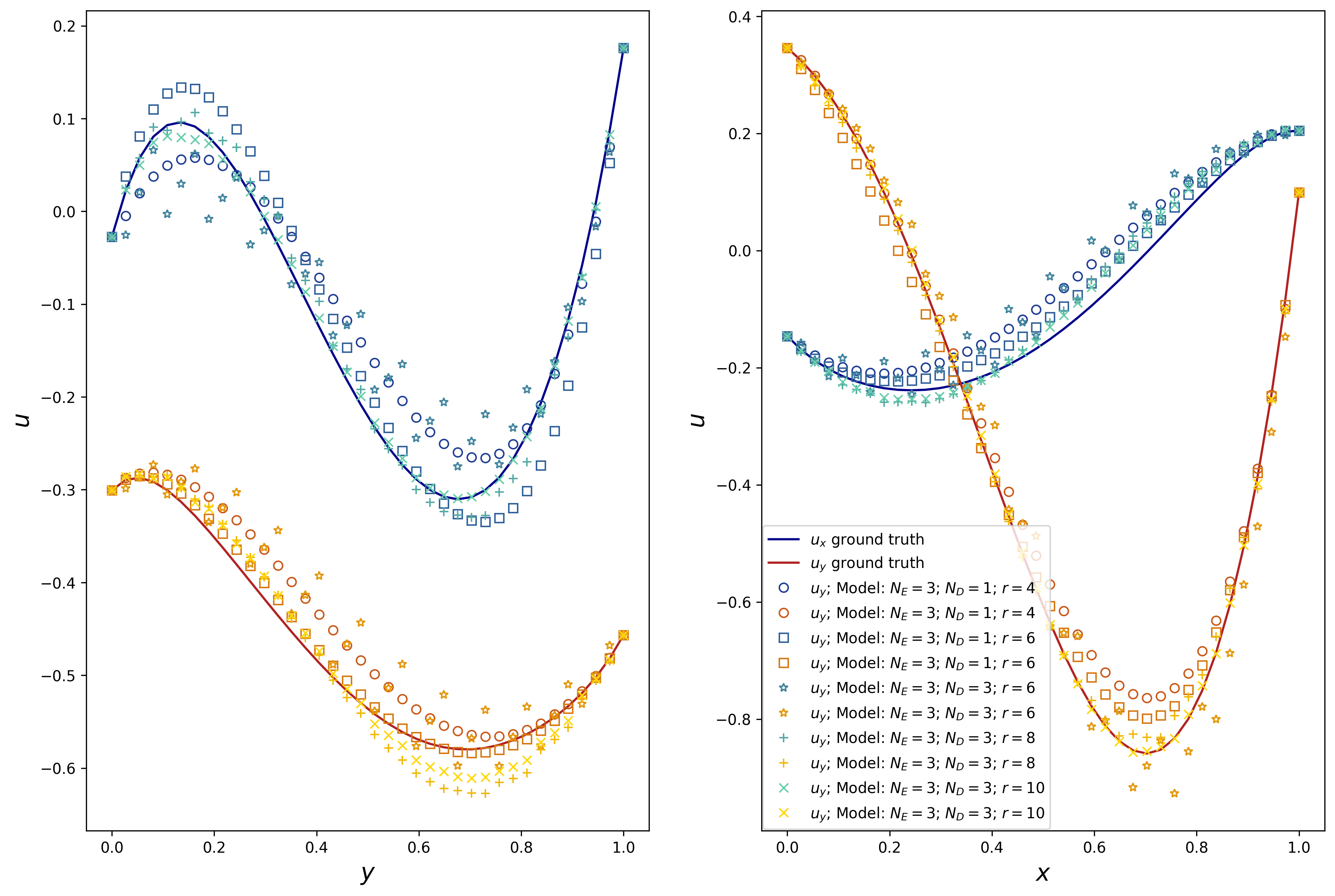}    
\caption{An indicative example of SINN performance for the parameters shown in Table \ref{tab: ResultsFluidExtraBC}. Results are plotted along slides of the domain where $x=0.5$ (left) and $y=0.5$ (right), with the original ``ground truth'' solution to the PDE \eqref{eq:LaminarFLow} also shown.}
\label{fig:ExtraBC_sliceplots_fluidU}
\end{figure}

\

Finally, we perform a latent variable sensitivity analysis for the two trained models with $(r=8,N_E=3,N_D=3)$ and $(r=10,N_E=3,N_D=3)$. The decoder sensitivities $\frac{\partial \delta}{\partial \bs{\ell}_{i}}(\overline{\bs{\ell}})$ each have two components which correspond to the two velocity components in the $x$ and $y$ directions. These sensitivities are shown in Figure \ref{fig:ExtradBC_sensitivity_fluid}. The model with $r=8$ exhibits approximately planar sensitivities about the mean latent variable value $\bar{\bs{\ell}}$, suggesting that the $r=8$ latent variables are being used by the trained model to enable planar perturbations to the solution in four directions for each of the two velocity components. Conversely, the model with $r=10$ clearly exhibits nonlinear, yet spatially coherent, latent-variable sensitivities which appear to enable a more accurate solution to the underlying boundary observation problem. 


\begin{figure}
\centering
\includegraphics[width=0.8\textwidth,clip=true,trim=0cm 0cm 0cm 0.05cm]{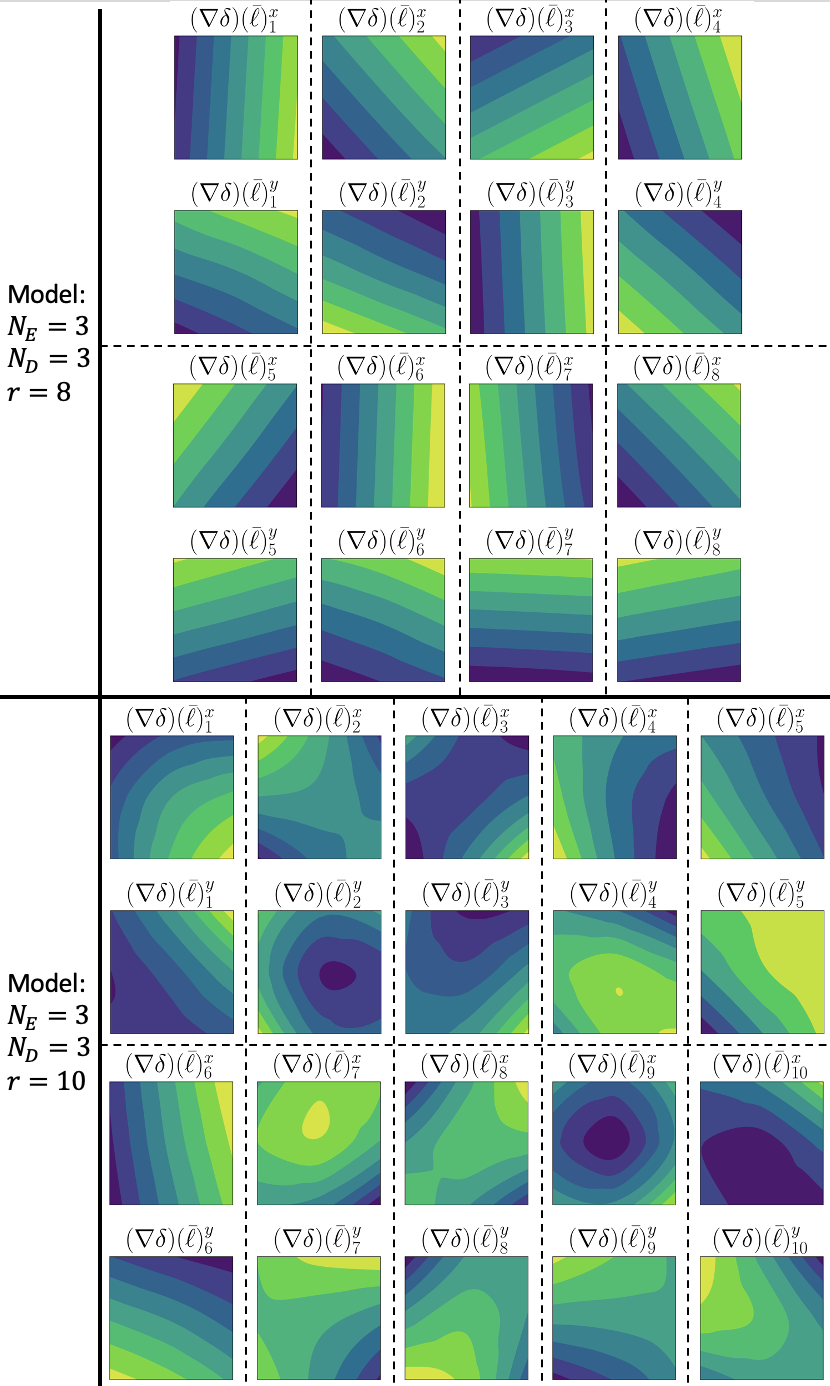}    
\caption{Latent variable sensitivities $(\nabla \delta)(\bar{\bs{\ell}})$ at the mean ensemble latent variable value $\bar{\bs{\ell}}$.}
\label{fig:ExtradBC_sensitivity_fluid}
\end{figure}

\section{Discussion}

The numerical examples discussed in \S\ref{sec:num_egs} suggest that SINNs are able to provide very good approximations to the nonlinear solution operators %
\[
\mathcal{F} :L^2(\partial \Omega, \mathbb{R}^{n_\partial}) \rightarrow L^2(\Omega, \mathbb{R}^n)
\]
for nonlinear boundary observation problems of the form \eqref{eq:PDE_abstract}. The power of the SINN approach is that, via training only finite-dimensional neural networks, it provides nonlinear infinite-dimensional operators $\mathcal{F}$ which can give approximate solutions $\mathcal{F}(\bs{b}) \in L^2(\Omega,\mathbb{R}^n)$ to a boundary observation problem for {\em any given} boundary data function $\bs{b} \in L^2(\partial \Omega,\mathbb{R}^{n_\partial})$. This represents a step-change in utility in comparison to data-driven approaches in which model training, and hence also the trained models, directly depends on a fixed instance of the boundary data. 

From the viewpoint of operator identification, the fact that ensemble errors in the range of $\mathcal{O}(10^{-3})$ to $\mathcal{O}(10^{-5})$ can be obtained by SINNs with very few latent variables ($3 \leq r \leq 10$) indicates that the approach has strong potential to be successfully applied to more complex examples. This is supported by the evidence, discussed in \S\ref{subsec:nl_heat}, that the semi-local structure of the SINN training algorithm endows the approach with significant robustness against over fitting. A further advantage of our data-driven approach is that SINNs can be obtained regardless of whether the available boundary data renders the underlying PDE boundary observation problem over- or under-determined. The data-driven operator $\mathcal{F}$ merely attempts to find an optimal approximation to the PDE solution, given the available training ensemble.  We also emphasise that SINN training does not require knowledge of the underling PDE, meaning that our method can be applied to experimental data and subsequently used to solve unseen boundary conditions.

A natural question is to ask whether the approximation error will converge to zero with increased complexity of the trained SINN operator (e.g. as $r, N_E \rightarrow \infty$, or with the complexity of the underlying neural networks). Since our aim is identify operators which solve nonlinear boundary observation problems which may have no closed-form solutions, and in view of the fact that linear elliptic systems are used as the central non-local building blocks of SINNs, it is unlikely that such convergence will hold in general. However, even without such a property, the numerical evidence presented in this paper suggests that SINNs can provide very good, low-complexity, approximations to nonlinear boundary observation problems which, furthermore, capture key physical features of the solution.

Viewing performance from an approximation accuracy philosophy is not out of line with the motivation for many well-established approaches to the simulation of complex nonlinear  PDEs. For example, in fluid mechanics, if one numerically solves the Reynolds Averaged Navier Stokes (RANS) equations, there is no expectation that the solution will agree with a fully resolved direct numerical simulation (DNS) of the governing Navier-Stokes equations. However, in many practical cases, a RANS solution may provide sufficiently physical insights at a substantially reduced computational cost than DNS. A similar philosophy applies to more accurate, yet still approximate, numerical approaches such as Large Eddy Simulation (LES). From the perspective of creating low-cost SINN models, it should be noted that there is technically no limit to using of significantly larger choice of extrapolation dimension $N_D$ than those used in the numerical examples considered in this paper. Furthermore, even if a SINN is trained using a finely-resolved spatial grid, the fact that an elliptic PDE is identified implies that the SINN operator $\mathcal{F} = \delta \circ \mathcal{E}_\mathcal{A} \circ \epsilon^\partial$ can be implemented using an arbitrary resolution, and potentially low-cost, solution to the central elliptic system $\mathcal{D}_\mathcal{A} \bs{\ell}=0$. 

The boundary encoding and internal decoding can be computationally expensive if large neural networks are used but unlike the latent elliptic system this computation is applied to each section of the domain independently and is trivial to parallelise. Given a discretisation grid of $n_i$ internal points the computational complexity of decoding would scale linearly with $\mathcal{O}(n_i)$ and the boundary encoding would scale even more favourably, since a typical choice of the number of boundary points $n_b$ is lower (e.g., for a 2D domain it may be assumed to scale as $n_b \propto \sqrt{n_i}$). On the other hand the latent elliptic system requires solving a linear system with $n_i \times r$ variables. Computational complexity of this step depends on specific algorithm, with direct methods such as Cholesky decomposition scaling as $\mathcal{O}((n_i r)^3)$. Since the SINN method only provides an approximate solution this precision is not required and so an iterative method could be used which has a smaller per iteration complexity. This is still significantly higher than the encoding or decoding step which, at large enough $n_i$, would dominate the computational cost. In summary, SINNs scale very well for large neural networks and, as shown in the section \S  \ref{subsec:nl_heat}, this can significantly increase modeling accuracy.

Finally, we comment on the computational cost of SINN training. A potential bottleneck is that, for each update to the elliptic system coefficients $\mathcal{A}$, one must repeatably solve a new elliptic system of PDEs on each training patch that is used to build up the cost function $\Psi(\mathcal{U},\mathcal{X})$. The cost of evaluating the cost function can be controlled by using a fixed number of training patch geometries $Q$, and by parallelising the elliptic system solutions on each training patch. To give an example, suppose that each training patch is as shown in Figure \ref{fig:training_run} and requires the solution of an elliptic system at $n_p$ internal points $\bs{p}_i$ in the training patch. Solution of this  elliptic system on the training patch domain involves solving a linear system: 
\[
A(\mathcal{A},\bs{p}_i,\bs{q}_{j})\bs{x} = \bs{b}(\mathcal{A},\bs{p}_i,\bs{q}_{j},\bs{\ell}(\bs{q}_{j}))
\]
where $A \in \mathbb{R}^{(n_p r) \times (n_p r)}$ is a symmetric positive definite matrix which depends linearly on the coefficients of $\mathcal{A}$, on the boundary points $\bs{q}_j$, and on the interior points $\bs{p}_{ij}$. The vector $\bs{b} \in \mathbb{R}^{n_p r}$ depends on $\mathcal{A}$, $\bs{q}_j$, $\bs{p}_{i}$ and the encoded boundary values $\bs{\ell}(\bs{q}_i)$. Solving the above linear system can be performed in two steps: forming the Cholesky decomposition $A=LL^\top$, then using $L$ to solve the linear system via  $\bs{x} = (LL^\top)^{-1} \bs{b}$.  If a common training patch geometry is used, the first step only needs to be computed once per training iterate, with the matrix $L$ stored in memory. 
This computation can be performed in parallel across all training patches required to compute $\Psi$. In a similar manner, any required evaluations of the encoder and decoder can also be parallelised. These steps imply that very efficient training of SINNs is possible.

\section{Conclusions} \label{sec:conculsions}

We have presented a data-driven method for solving boundary observation problems which identifies a solution operator which can approximate the PDE solution for arbitrary boundary data. The constructed models, referred to here as Structure Informed Neural Networks (SINNs), embed an elliptic system into a classical encoder/decoder Neural-Network architecture for reduced-order modelling. The use of elliptic systems, which are well-posed with respect to the global passage of problem data, enables very efficient model training to be performed on small patches of the underlying domain. Numerical evidence suggests that this endows the proposed SINN methodology with significant robustness to over-fitting.

The methodology presented in this paper can be used to solve boundary observation problems which are both time-independent and have boundary data which is known on the entire boundary. Future research will investigate the possibility of extending the SINN methodology to handle cases in which only partial boundary data is available for training or testing, the potential for SINN operators to be embedded in time-dependent algorithms for boundary observation, and the application of the developed methodology to more complex domain geometries.

\section{Appendix}  \label{sec:appendix}

We present the proofs of the regularity results stated in the paper. 

\subsection{Proof of Lemma \ref{lem:cts_int_encoder}} \label{app:int_encoder}

{\em Regularity of $\epsilon \bs{u}$:} Given $\bs{x},\bs{y} \in \Omega_E$, note that 
\begin{equation} \label{eq:eu}
|(\epsilon \bs{u})(\bs{x}) - (\epsilon \bs{u})(\bs{y})| = |e(\bs{u}_{\bs{x}}) - e(\bs{u}_{\bs{y}})|.
\end{equation}
Now, if $\bs{x} \rightarrow \bs{y}$ in $\Omega_E$, then by a standard approximation argument, $\|\bs{u}_{\bs{x}} - \bs{u}_{\bs{y}}\|_{L^2(E)} \rightarrow 0$. It then follows from \eqref{eq:eu} and the assumed continuity of the generating function $e$ that $(\epsilon \bs{u})(\bs{x}) \rightarrow (\epsilon \bs{u})(\bs{y})$, meaning that $(\epsilon \bs{u}) : \Omega_E \rightarrow \mathbb{R}^r$ is continuous. 

To prove uniform boundedness of $\epsilon \bs{u}$, note that for for any $\bs{u} \in L^2(\Omega)$, 
\[
\sup_{x \in \Omega_E} \|\bs{u}_{\bs{x}} \|_{L^2(E)}^2 = \sup_{x \in \Omega_E} \int_E |u(\bs{x}+\bs{y})|^2 d\bs{y} \leq \|\bs{u}\|^2_{L^2(\Omega)}
\]
Since $e : L^2(E) \rightarrow \mathbb{R}^r$ is compact,  it maps bounded subsets of $L^2(E)$ to bounded subsets of $\mathbb{R}^r$. Hence, 
\[
\sup_{x \in \Omega_E} |(\epsilon \bs{u})(\bs{x})| = \sup_{x \in \Omega_E} |e(\bs{u}_{\bs{x}})|_2 < \infty. 
\]
Consequently, $\epsilon \bs{u} \in C(\Omega_E,\mathbb{R}^r)$. \qed

\subsection{Proof of Lemma \ref{lem:decode_cts}} \label{sec:app_decode}

Let $\bs{\ell} \in C(\Omega,\mathbb{R}^r)$ and let $\epsilon >0$. Let $\bs{x},\bs{z} \in \Omega$ and define sets $D_{\bs{x}\bs{y}}:=( D_{\bs{x}} \cap D_{\bs{z}} \cap \Omega)$ and
\[
D_{\bs{x} \setminus \bs{z}} = (D_{\bs{x}} \cap \Omega) \setminus D_{\bs{x}\bs{y}}, \quad D_{\bs{z} \setminus \bs{x}} = (D_{\bs{z}} \cap \Omega) \setminus D_{\bs{x}\bs{y}} 
\]
and set volumes by
\[
c_{\bs{x}} = |D_{\bs{x}} \cap \Omega|, \quad c_{\bs{z}} = |D_{\bs{z}} \cap \Omega|.
\]
For convenience, we also let $f(\cdot):=(\delta \bs{\ell})(\cdot)$ and $g_{\bs{y}}(\cdot):= d(\ell(\bs{y}))(\cdot)$. Then, 
\begin{align*}
|f(\bs{x})-f(\bs{z})| &= \left|\frac{1}{c_{\bs{x}} }\int_{D_{\bs{x}}} g_{\bs{y}}(\bs{x}-\bs{y}) d\bs{y} - \frac{1}{c_{\bs{z}} } \int_{D_{\bs{z}}} g_{\bs{y}}(\bs{z}-\bs{y}) d\bs{y}\right| \\
&\leq \underbrace{\frac{1}{c_{\bs{x}} }\int_{D_{\bs{x}\setminus \bs{z}}} |g_{\bs{y}}(\bs{x}-\bs{y})| d\bs{y} + \frac{1}{c_{\bs{z}} }\int_{D_{\bs{z}\setminus \bs{x}}} |g_{\bs{y}}(\bs{z}-\bs{y})| d\bs{y}}_{:=I_1}\\
&\quad + \underbrace{\frac{1}{c_{\bs{x}} }\int_{D_{\bs{x} \bs{z}}} \left| g_{\bs{y}}(\bs{x}-\bs{y}) - g_{\bs{y}}(\bs{z}-\bs{y}) \right| d\bs{y}}_{:=I_2}\\
&\quad + \underbrace{\left| \frac{1}{c_{\bs{x}}} - \frac{1}{c_{\bs{z}}} \right| \int_{D_{\bs{x}{\bs{z}}}} | g_{\bs{y}}(\bs{z}-\bs{y})| d\bs{y}}_{:=I_3}
\end{align*}
Now, since $\bs{\ell} \in C(\Omega,\mathbb{R}^r)$, it follows that $\bs{\ell}(\Omega) \subset \mathbb{R}^r$ is bounded. Then, using compactness of the encoder generating function  $d$, it follows that $\{ d(\bs{\ell})(\bs{y}) \}_{\bs{y} \in \Omega} = \{ g_{\bs{y}} \}_{\bs{y} \in \Omega}$ is a bounded subset of $C(\Omega,\mathbb{R}^n)$. Hence, there exists $K>0$ such that 
\begin{equation} \label{eq:uni_bnd}
\sup_{\bs{y} \in \Omega} \|g_{\bs{y}}\|_{C(\Omega,\mathbb{R}^n)} \leq K < \infty. 
\end{equation}
Then, since $D_{\bs{x} \setminus \bs{z}},D_{\bs{z} \setminus \bs{x}} \rightarrow 0$ and $c_{\bs{x}} - c_{\bs{z}} \rightarrow 0$ as $\bs{x} \rightarrow \bs{z}$, it follows that there exists $\delta_1 >0$ such that  
\[
I_1 + I_3 \leq K \left( |D_{\bs{x} \setminus \bs{z}}| + |D_{\bs{z} \setminus \bs{x}}| + |D_{\bs{x}\bs{z}}|\left| \frac{1}{c_{\bs{x}}} - \frac{1}{c_{\bs{z}}} \right| \right) < \frac{\epsilon}{2}. 
\]
whenever $|\bs{x}-\bs{y}| < \delta_1$. 

Finally, since $d$ is continuous and $\bs{\ell}(\bar{\Omega}) \subset \mathbb{R}^d$ is compact, it follows that $\mathcal{F}:=\{ g_{\bs{y}} \}_{\bs{y} \in \bar{\Omega}}$ is a compact subset of $C(\Omega,\mathbb{R}^n)$. Consequently, the set of functions $\mathcal{F}$ is equicontinuous and, hence, there exists $\delta_2>0$ such that $|\bs{g}_y(\bs{x}-\bs{y}) - g_{\bs{y}}(\bs{z}-\bs{y})| < \epsilon/2$, for any $\bs{y} \in \Omega$, whenever $|\bs{x}-\bs{z}| <\delta_2$. Hence, 
\[
|f(\bs{x})-f(\bs{z})| \leq I_1 + I_2+I_3 \leq \epsilon
\]
whenever $|\bs{x}-\bs{z}| <\min\{\delta_1,\delta_2\}$, meaning that $f = \delta \bs{\ell}$ is continuous. That $\delta\bs{\ell} \in C(\Omega,\mathbb{R}^n)$ then follows from the upper bound \eqref{eq:uni_bnd}.
\qed

\bibliography{SINN_references}

\end{document}